\documentclass[a4paper]{article}

\usepackage[english]{babel}
\usepackage[utf8x]{inputenc}
\usepackage[T1]{fontenc}

\usepackage{authblk}
\usepackage[hidelinks]{hyperref}

\usepackage{amsmath, amsthm, amsfonts, mathtools}

\newtheorem{lemma}{Lemma}[section]

\usepackage{graphicx}
\usepackage{subcaption}
\usepackage{float}
\usepackage{booktabs}

\newtheorem{theorem}{Theorem}

\usepackage[a4paper,top=3cm,bottom=2cm,left=3cm,right=3cm,marginparwidth=1.75cm]{geometry}
\usepackage{authblk}

\usepackage{amsmath}
\usepackage{graphicx}
\usepackage[colorinlistoftodos]{todonotes}
\usepackage[authoryear,round]{natbib}
\title{Travel Time Reliability in Stochastic Kinematic Wave Models}

\author[1]{Alexander Hammerl\thanks{Corresponding author: \texttt{asiha@dtu.dk}}}
\author[1]{Ravi Seshadri}
\author[1]{Thomas Kjær Rasmussen}
\author[1]{Otto Anker Nielsen}

\affil[1]{Department of Transport, Technical University of Denmark}

\date{\vspace{-5ex}}

\begin{document}
\maketitle
\renewcommand{\abstractname}{Abstract}

\begin{abstract}
This paper analyzes the time-dependent relationship between the mean and variance of travel time of vehicular traffic on a single corridor under rush hour like congestion patterns. To model this phenomenon, we apply the LWR (\cite{Lighthill1955}, \cite{Richards1956}) theory on a homogeneous freeway with a discontinuous bottleneck at its downstream end, assuming a uni-modal demand profile with a stochastic peak. We establish conditions for typical counterclockwise hysteresis loops under these assumptions and provide a general mathematical characterization of systems which exhibit counterclockwise dynamical behavior. It is demonstrated that shapes of the fundamental diagram which produce a counterclockwise loop can be interpreted as an indication of aggressive driving behavior, while deviations may occur under defensive driving. This classification enables an identification of the qualitative physical mechanisms behind this pattern, as well as an analysis of the causes for quantitatively limited deviations. An empirical validation using loop detector data from the PeMS dataset supports the prevalence of counterclockwise loops under typical rush hour conditions. The obtained results improve the understanding of the causes of this hysteresis pattern and of the properties of kinematic flow models under stochastic boundary conditions.
\end{abstract}

\textbf{Keywords}: traffic flow theory, kinematic waves, hyperbolic conservation laws, hysteresis.

\section{INTRODUCTION}
Travel time reliability is a critical aspect of traveler route choice in urban areas, with empirical analyses showing it can be nearly as important to travelers as the expected travel time itself (e.g. \cite{Li2010,Prato2014,seshadri2017robust,prakash2018consistent}). Travel time variability provides
a convenient measure of reliability and can be analyzed in different reference frames: vehicle-to-vehicle, period-to-period and day-to-day. To model day-to-day variability, which is the focus of this article, a linear relationship between mean and variance of travel time is often assumed (e.g. \cite{Kim2015}, \cite{VanLint2008}). However, empirical data show an anti-clockwise hysteresis loop between these two quantities, both at the level of individual links (\cite{Fosgerau2010}, \cite{Kim2015}, \cite{Yildirimoglu2015}) and at a network level (\cite{Bates2004}, \cite{Gayah2015}). There have been only a few attempts to explain the nature of this relationship and hysteresis pattern. Fosgerau \cite{Fosgerau2010} theoretically proves the existence of such loops in a queueing system with decreasing arrival rate. Yildirimoglu et al. \cite{Yildirimoglu2015} attribute the hysteresis in the day-to-day travel time variability to stochastic parameters of vehicle travel time, but do not analyze which aspects of traffic flow cause or influence this effect. \cite{Gayah2015} demonstrates that travel time variability exhibits counterclockwise hysteresis in networks with clockwise hysteresis in its macroscopic fundamental diagram (MFD).

Separate from the hysteresis in travel time variability are other phenomena in traffic research that bear the same name: Edie \cite{Edie1963} and Treiterer and Myers \cite{Treiterer1974} define hysteresis as the separation of speed-density curves into an accelerating and a decelerating branch ahead of traffic disturbances. Zhang \cite{Zhang1999} and Yeo and Skabradonis \cite{Yeo2009}, among others, offer theoretical explanations for this effect. Due to the high relevance and operational importance of the MFD as a modeling tool, an extensive discussion has developed in the literature regarding the conditions under which MFD curves are well-defined and when they exhibit hysteresis effects (\cite{Geroliminis2008,Daganzo2008,Helbing2009,Geroliminis2011, Geroliminis2011b, Rempe2016}). Various factors driving hysteresis loops have been identified across different studies. For instance, \cite{Buisson2009, Mazloumian2010} emphasize spatial heterogeneity in congestion distribution during different traffic phases as a primary mechanism, while \cite{Buisson2009,Yi2010,mahal13,leqal15} demonstrate that highway-like network topologies with limited route choice options play an equally important role. Network instabilities triggered by exceptional events represent another significant contributor \cite{Gayah2011}. Building on these findings, \cite{Gayah2015} show that in networks exhibiting MFD hysteresis, fluctuations in macroscopic flow quantities can additionally lead to hysteresis effects in the mean-variance plane.

Although MFD hysteresis has been extensively discussed in the literature, hysteresis associated with the mean and variance of travel time (in the context of day-to-day variability) has received relatively little attention. \cite{Fosgerau2010} demonstrates that at least one counterclockwise loop occurs in a simplified queue-based model of rush hour traffic. However, the precise shape of the dynamic relationship remains unresolved in their analysis: "It is possible for the curve to move to the North-East, back towards the origin and North-East again, after turning clockwise some to the West of the first extreme point to the East" (\cite{Fosgerau2010}, p.5). The model developed in the present work comprises \cite{Fosgerau2010}'s model as a special case. We prove that the subloops described above cannot occur, establishing that a simple counterclockwise loop is the only possible trajectory. Furthermore, we contribute to the existing literature by investigating the time-dependent relationship between mean and variance of travel time in a single corridor under rush hour specific traffic dynamics using LWR theory. We analytically derive conditions for the occurrence of anti-clockwise hysteresis loops and provide a rigorous theoretical foundation to explain which traffic flow variables determine the shape and size of the hysteresis loop. The theoretical analysis is supported by empirical data from Interstate 880 in San Francisco.

\section{Model and General Solution}
The LWR theory is a macroscopic model that describes traffic flow on highways using the variables \textit{traffic density} \(k\) (the number of vehicles per unit length) and \textit{traffic flow} \(q\) (the number of vehicles passing a point per unit time). It asserts that the rate of change in the total number of vehicles contained in any road segment \( [x_1, x_2] \) where $x_2>x_1$ and there are no entries or exits,  is equal to the net flow of vehicles out of the segment, i.e.
\begin{equation}
\frac{d}{dt} \int_{x_1}^{x_2} k(x,t) \, dx = - \left[ q(x,t) \right]_{x_1}^{x_2},
\label{eq:integralconservation}
\end{equation}
If \( k \) and \( q \) are differentiable functions, we can, on the left, perform the differentiation under the integral sign and on the right apply the fundamental theorem of calculus to obtain 
\[
\int_{x_1}^{x_2} \left\{ \frac{\partial k}{\partial t} + \frac{\partial q}{\partial x} \right\} dx = 0.
\]
Since this relation is valid for any arbitrary road segment $[x_1, x_2]$, by letting $x_2 \to x_1$ and dividing by the segment's length, the expression simplifies to the partial differential equation
\begin{equation}
	\frac{\partial k}{\partial t} + \frac{\partial q}{\partial x} = 0.
	\label{eq:conservation}
\end{equation}
In addition to \ref{eq:integralconservation}, the LWR theory assumes the existence of a functional relationship between $q$ and $k$ under differentiable conditions:
\begin{equation}
	q(x,t)=Q(x,k(x,t)),
	\label{eq:fundamental}
\end{equation}
where $Q$ is a concave, non-negative function that is equal to zero at $k=0$ and at the \textit{jam density} $k=k_j$.
On substituting equation \ref{eq:conservation} into \ref{eq:fundamental}, we obtain 
\begin{equation}
	\frac{\partial k}{\partial t} + \frac{dQ}{dk} \cdot \frac{\partial k}{\partial x} = 0,
	\label{eq:differentiallwr}
\end{equation}
which defines a unique solution for $k(x,t)$ and $q(x,t)$ for given initial and boundary conditions if $q$ and $k$ are differentiable. We denote by $k_{\text{crit}}$ the density that maximizes the fundamental flow-density relationship, i.e., $k_{\text{crit}} = \arg\max_k q(k)$.
Flow and density are related to the cumulative flow \(N(x,t)\) as follows:

\begin{equation}
	q(x,t) = \frac{\partial N}{\partial t} (x,t), \quad k(x,t) = -\frac{\partial N}{\partial x} (x,t)
\end{equation}

In cases where \(k\) has a discontinuity at \((x,t)\), known as a shockwave, the shockwave's speed \(u\) is specified as:
\begin{equation}
	u = \frac{[q]}{[k]} = \frac{q_2 - q_1}{k_2 - k_1}.
	\label{eq:rhjump}
\end{equation}

It should be emphasized that customary solutions of the LWR model do \textbf{not} solve the system of equations \ref{eq:integralconservation} and \ref{eq:fundamental}. Instead, they solve equation \ref{eq:differentiallwr}, which assumes the fundamental diagram holds only where traffic variables are differentiable in $x$ and $t$ (not at shocks). For solutions at discontinuities, additional entropy conditions may be required to ensure uniqueness. For details on entropy solutions in traffic flow models, see for example \cite{Lebacque1996}, \cite{Ansorge1990}, \cite{Jin2009}.

However, entropy solutions are not necessary for our analysis. We study a homogeneous version of the LWR model where the fundamental diagram remains constant throughout the spatial domain, except at the downstream bottleneck: 
\begin{equation}
    q(k(x,t),x) = q(k(x,t)) \quad \text{for all } x \in [0, l).
\end{equation}
Moreover, our analysis considers only continuous boundary conditions. In this setting, every point in spacetime is intersected by at least one kinematic wave, allowing us to uniquely determine $N(x,t)$ \cite{Newell1993, Newell1993b, Daganzo2005}. When a solution exists for $q(x,t)$ and $k(x,t)$, it is also unique.

For the purposes of our analysis, we need only assume that equation \ref{eq:differentiallwr} holds where $k$ and $q$ are differentiable, and that equation \ref{eq:integralconservation} holds at all points. This framework is sufficient to derive all results presented in this paper.

The third variable, the \textit{average speed}, is defined as $v=\frac{q}{k}$. We model aggressive driving by assuming \( v''(k) \leq 0 \), while defensive driving is characterized by \( v''(k) \geq 0 \). This implies that defensive drivers reduce their speed following a density increase even in light traffic, whereas aggressive drivers maintain higher speeds until they are closer to the jam density.

Traffic moves along a road segment of length $l$, which ends in a bottleneck with a maximum capacity $q_{bn}$. The travel time $\tau(t)$ for a vehicle entering the segment at time $t$ is described by:
\begin{equation}
	\tau(t) = \inf \{ T \geq 0 : N(l, t + T) > N(0, t) \},
\end{equation}

It is assumed that the length of the queue never exceeds the physical space available on the road segment. To model rush hour traffic, the upstream boundary flow $q(0,t)$ is represented as a trapezoidal function with a randomly distributed peak $q_p \sim \phi$, defined as:

\begin{equation}
	q(0,t) = 
	\begin{cases} 
		q_b + a \cdot t, & \text{for } 0 \leq t \leq \frac{q_p-q_b}{a}, \\
		q_p, & \text{for } \frac{q_p-q_b}{a} \leq t \leq t_{pe}, \\
		q_p - b \cdot (t_{e} - t), & \text{for } t_{pe} \leq t \leq \frac{q_p}{q_e \cdot b}+t_{pe}, \\
		q_e, & \text{for } \frac{q_p}{q_e \cdot b}+t_{pe} \leq t \leq \infty.
	\end{cases}
\end{equation}

for suitably chosen parameters \( q_b \) (initial flow), \( q_e \) (end flow), $t_{pe}$ (end time of peak congestion), \( a \) (flow increase rate at the onset of congestion), and \( b \) (flow reduction rate at the offset of congestion). We further assume that $q_e<q_{bn}$, so that the expected travel time for very late departure times, as $t_{\text{dep}} \to \infty$, approaches the free flow travel time $\tau_{\text{free}}=\frac{l}{v(0)}$. The travel time of a vehicle departing at time $t$, given that the peak boundary flow is $q_p$, is denoted as $\tau(t,q_p)$. 

Analytical solutions of the LWR model often rely on the method of characteristics, which describes the conservation of density \( k(x,t) \) along a path with location-dependent wave speed \( q'(k) \). Under smooth conditions, the density remains constant along this path:

\[
\frac{dk}{dt} = \nabla k \cdot \begin{bmatrix} q'(k) \\ 1 \end{bmatrix} = \frac{\partial k}{\partial x} q'(k) + \frac{\partial k}{\partial t} = 0,
\]

where the second equation follows from the conservation law \ref{eq:conservation}. If a point in space-time is traversed by more than one characteristic, a shock wave forms at that point. The works by \cite{Newell1993}, \cite{Daganzo2003}, and \cite{Daganzo2005} provide the following rule to ensure the physical correctness of individual characteristics:

\begin{lemma}[\cite{Newell1993}, \cite{Daganzo2005}]
\label{lemma:newell}
The integral conservation law \ref{eq:integralconservation} is satisfied at a point \( (x, t) \) by the characteristic that intersects this point associated with the lowest cumulative flow \( N(x,t) \).
\end{lemma}

While \cite{Daganzo2005b} addresses the computational efficiency of identifying valid characteristics, the following lemma provides an even more efficient method for our specific problem:

\begin{lemma}
\label{lemma:shockpriority}
Let $q'(k)$ be a concave function in k. Then, for every point \((x,t)\) satisfying \( x < \psi(t) \) which is reached by at least one characteristic curve, the physically correct characteristic is the latest emanating one.
\end{lemma}
\begin{proof}[Proof (Sketch)]
The lemma is proven by analyzing a discretized approximation of the upstream boundary condition. By the concavity of \(q(k)\), characteristics may intersect only if they originate from the descending part of the boundary. We partition the decreasing branch into intervals \(I_1, \ldots, I_n\). 
Due to the concavity of \( q'(k) \), the characteristic speed \( q'(k) \) increases more strongly between \( I_{n-1} \) and \( I_n \) than between \( I_n \) and \( I_{n+1} \). Consequently, the intersection point of the lines starting in \( I_{n-1} \) and \( I_n \) occurs before the intersection point of the lines starting in \( I_n \) in \( I_{n+1} \).
Define \(q_{\text{discr}}(t)\) as \(q(0, t_I)\) where \(t_I\) is the lower bound of the interval \(I\) containing \(t\). Additionally, we linearize \(q(k)\) over the decreasing branch. Suppose two characteristic lines intersect at \((x, t)\), with \(c_1\) from \(I_1\) and \(c_2\) from \(I_2\), \(I_1 < I_2\), and \(c_2\) is most recent. \(c_1\) must have crossed a shockwave, representing the physically valid solution at this point in space-time. The continuous boundary condition solution derives from this discretization method as intervals approach zero length.
\end{proof}

Figure~\ref{fig:approximation_method} illustrates the approximation method of the lemma. Figure~\ref{fig:transformation} displays the transformation of a continuously decreasing boundary flow (green) into three (blue) or six (orange) discrete steps. Figures~\ref{fig:solution_a} and~\ref{fig:solution_b} show the resulting solutions. These steps propagate as shock waves. The characteristics that intersect the point $(\frac{4}{3}, 25)$ are shown as dotted lines. It is straightforward to verify graphically that only the later emanating characteristic represents a feasible solution of the LWR theory in both cases.

\begin{figure}[H]
	\centering
	\begin{subfigure}{.5\textwidth}
		\centering
		\includegraphics[height=0.7\linewidth]{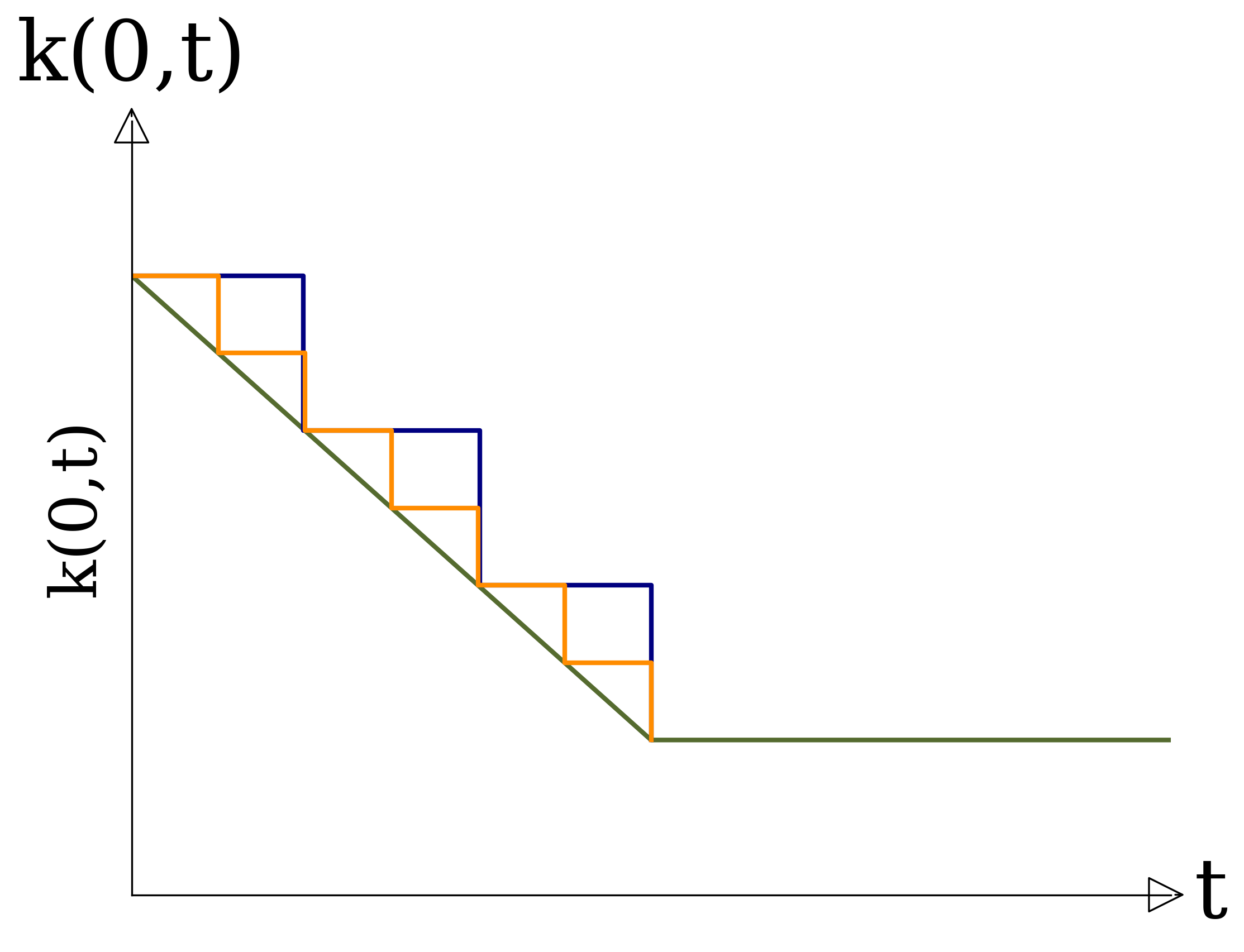} 
		\caption{Discretized flow}
		\label{fig:transformation}
	\end{subfigure}%
	\hfill
	\begin{subfigure}{.5\textwidth}
		\centering
		\includegraphics[height=0.7\linewidth]{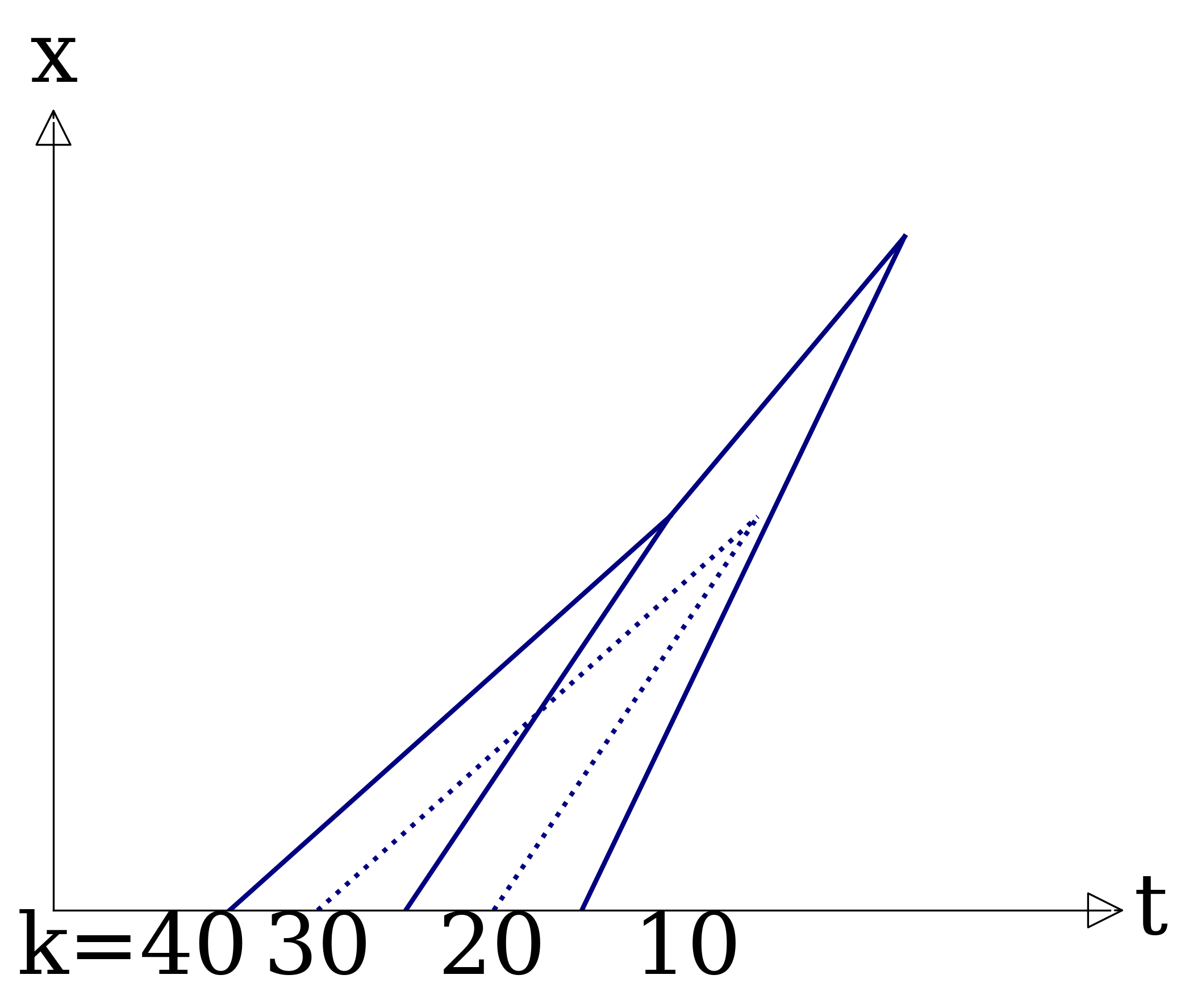} 
		\caption{Approximate solution 1}
		\label{fig:solution_a}
	\end{subfigure}%
	\hfill
	\begin{subfigure}{.5\textwidth}
		\centering
		\includegraphics[height=0.7\linewidth]{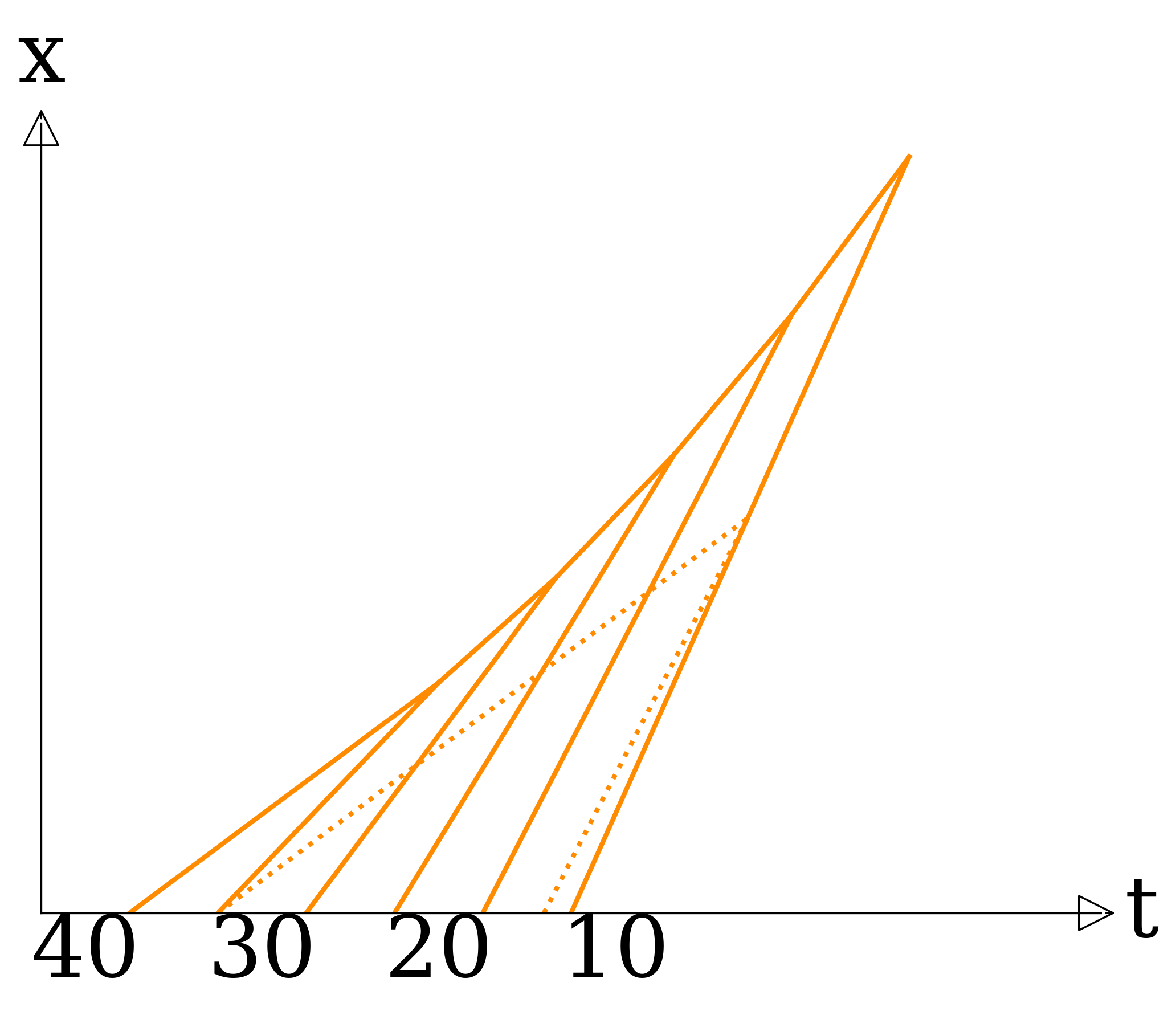} 
		\caption{Approximate solution 2}
		\label{fig:solution_b}
	\end{subfigure}
	\caption{Visualization of the discretized flow and approximate solutions}
	\label{fig:approximation_method}
\end{figure}

\section{The Shape of the Mean-Variance Curve}
\label{section:shape}

We consider two vehicles, A and B, entering a traffic corridor at times $t_1$ and $t_2$ respectively, where $t_1 < t_2$. Both vehicles have identical expected travel times which implies that $t_1 \leq t_{pe} \leq t_2$.
Throughout this section, we assume aggressive driving behavior on the uncongested branch of the fundamental diagram, i.e., $v''(k) \leq 0$ for $ k \leq k_{\text{crit}}$, which implies $q^{(3)}(k) \leq 0$ and $k^{(3)}(q) \geq 0$ in this interval.

Let us define $q_p^-$ and $q_p^+$ such that $\tau(t_1, q_p^-)=\tau(t_2, q_p^-)$ and $q_p^+>q_p^-$. We denote the arrival times of both vehicles at these peak flow values as \( t_{1,a}^- \) and \( t_{2,a}^- \), and \( t_{1,a}^+ \) and \( t_{2,a}^+ \), respectively. For analytical convenience, we normalize the corridor length to 1 without loss of generality of the presented proof and set the time scale such that $b=1$ applies at a peak flow of $q_p^+$. The extension of our results to arbitrary units of space and time can then be achieved by simply scaling the relevant terms in the proof by the factors $l$ and $b$, respectively.

To assess the difference in travel time variance, we consider travel time difference $\Delta\tau\coloneqq\tau(t_2)-\tau(t_1)$ as a function of $q_p$, for which the following relationship holds: 

\begin{lemma}
	\label{lemma:characterization}
	Assume there exists an \( x_0 \) such that \( \Delta \tau(x) < 0 \) for \( x < x_0 \) and \( \Delta \tau(x) > 0 \) for \( x > x_0 \), where \(\tau_1(q_p) > 0\) and \(\tau_1\) is an increasing function of \( q_p \). Then, the following inequality holds:
	\[
	\operatorname{Var}[\tau_2] \geq \operatorname{Var}[\tau_1].
	\]
\end{lemma}

\begin{proof}
\begin{align*}
	\mathbb{E}[\tau_1 \Delta \tau] &= \int \tau_1(\Delta \tau \phi(x)) \, dx \\
	&= \int_0^{x_0} \tau_1 \Delta \tau \phi(x) \, dx + \int_{x_0}^\infty \tau_1 \Delta \tau \phi(x) \, dx \\
	&\geq \tau_1(x_0) \cdot \int_0^{x_0} \Delta \tau \phi(x) \, dx + \tau_1(x_0) \cdot \int_{x_0}^\infty \Delta \tau \phi(x) \, dx \\
	&= \tau_1(x_0) \cdot \mathbb{E}[\Delta \tau]= 0.
\end{align*}

\begin{align*}
	\operatorname{Var}(\tau_2) &= \operatorname{Var}(\tau_1 + \Delta \tau) \\
	&= \operatorname{Var}(\tau_1) + \operatorname{Var}(\Delta \tau) + 2\operatorname{Cov}(\tau_1, \Delta \tau) \\
	&= \operatorname{Var}(\tau_1) + \operatorname{Var}(\Delta \tau) + 2(\mathbb{E}[\tau_1 \Delta \tau] - \mathbb{E}[\tau_1] \mathbb{E}[\Delta \tau]) \\
	&= \operatorname{Var}(\tau_1) + \operatorname{Var}(\Delta \tau) + 2\mathbb{E}[\tau_1 \Delta \tau] \ge \operatorname{Var}(\tau_1).
\end{align*}
\end{proof}

The expected travel time is a uni-modal function of the departure time and approaches the minimum travel time for both very early and very late departures. Therefore, if \(\operatorname{Var}[\tau_2] \geq \operatorname{Var}[\tau_1]\), this corresponds to a counterclockwise movement of the curve in the $(\mathbb{E}, \operatorname{Var})$ plane.

For the remainder of the proof, we define \(\Delta \tau_f\) as the difference in travel time between A and B, if no queue is active at the arrival time of either vehicle. Similarly, we define \(\Delta \tau_q\) as the difference in travel time between A and B if a queue is active at the arrival time of both vehicles. We first establish that, at peak flow $q_p^+$,  $\Delta \tau_q$ and $\Delta \tau_f$ are non-negative. This implies a single sign change from negative to positive, thus satisfying the condition of Lemma \ref{lemma:characterization}. We then demonstrate that this property is preserved in the combined residual delay $\Delta \tau=\tau(t_2)-\tau(t_1)$.

In the following section, we present a three-part proof to demonstrate by case distinction that the conditions of Lemma \ref{lemma:characterization} are satisfied for the following classes of peak flows: 
(i) both vehicles encounter an active queue (Lemma \ref{lemma:c_final}), 
(ii) no vehicle encounters a queue (Lemma \ref{lemma:f_final}), 
(iii) exactly one vehicle encounters a queue (Lemma \ref{lemma:fc_final}). 
The application of Lemma \ref{lemma:characterization} then leads to the main result of the paper:

\begin{theorem}
\label{theorem:main}
Consider aggressive driving behavior, characterized by a concave speed-density relationship in the free-flow regime (i.e., $v''(k) \leq 0$ for densities $k \leq k_{\text{crit}}$). Under these conditions, the path described by the expected value $\mathbb{E}[\tau(t)]$ and variance $\operatorname{Var}[\tau(t)]$ forms a counterclockwise loop in the mean-variance phase plane.
\end{theorem}

We note that demonstrating the existence of a single counterclockwise loop requires the aggressive driving property only in the case of free-flow delays. Stochastic queueing delays always produce a counterclockwise loop, regardless of whether the aggressive driving property is satisfied.

To improve readability, most of the proofs of the lemmas presented in the rest of this section have been provided in the appendix. 
The main symbols and variables used throughout the text are summarized in Table \ref{tab:notation}.

\begin{table}[h!]
\centering
\begingroup
\renewcommand{\arraystretch}{1.2}
\hspace{-1 cm}
\begin{tabular}{@{}ll@{}}
\toprule
\textbf{Symbol} & \textbf{Meaning (units)}\\
\midrule
$x \in [0,l]$ & Longitudinal position along the corridor (length)\\
$t \ge 0$ & Time (time)\\
$l$ & Corridor length (length)\\[2pt]

$k(x,t)$ & Traffic density (veh/length)\\
$q(x,t)$ & Traffic flow (veh/time)\\
$v(x,t)$ & Space–mean speed; $v = q/k$ where defined (length/time)\\
$Q(k)$ & Fundamental diagram (flow–density relation), concave, $Q(0)=Q(k_j)=0$\\
$k_j$ & Jam density (veh/length)\\
$k_c$ & Critical density (arg\,max of $Q$)\\
$v_f$ & Free-flow speed (length/time)\\[2pt]

$N(x,t)$ & Cumulative count/Moskowitz function (veh)\\
$\tau(t)$ & Travel time of a vehicle entering at time $t$; $\tau(t)=\inf\{T\!\ge\!0:\, N(l,t+T)>N(0,t)\}$ (time)\\
$\Delta \tau$ & Difference in travel times between scenarios/vehicles (time)\\
$\Delta \tau_q$ & Travel time difference due to queuing delays (time)\\
$\Delta \tau_f$ & Travel time difference due to free-flow delays (time)\\
$\Delta \tau_{f,q}$ & Travel time difference when one vehicle is in free flow and the other in a queue (time)\\[2pt]

$q(0,t)$ & Upstream boundary inflow (veh/time)\\
$q(l,t)$ & Downstream boundary discharge (veh/time)\\
$\bar q(\cdot)$ & Average flow (veh/time) (precise definition given where used)\\[2pt]

$q_b$ & Base (pre-peak) inflow level (veh/time)\\
$q_p$ & Peak boundary flow (veh/time)\\
$q_e$ & Post-peak/end inflow level (veh/time)\\
$a$ & Morning ramp-up rate of inflow (veh/time$^2$)\\
$b$ & Evening ramp-down rate of inflow (veh/time$^2$)\\
$t_{pb}$ & Start time of the peak inflow (time)\\
$t_{pe}$ & End time of the peak inflow (time)\\
$t_p$ & Peak time (time)\\[2pt]

$A,B$ & Labels for two probe vehicles considered in comparisons\\
$\tau_A,\ \tau_B$ & Travel times of vehicles $A$ and $B$ (time)\\
$t_1,\ t_2$ & Entry times of vehicles $A$ and $B$ (time)\\[2pt]
$t_{c,A},\ t_{c,B}$ & Departure times of the characteristic curves whose arrivals coincide with vehicles $A$ and $B$ (time)\\[2pt]

$s_1,\ s_2$ & Characteristic curves used in the analysis (definition given in text)\\[2pt]

$q^*(\cdot),\ v^*(\cdot)$ & Counterfactual/auxiliary flow– and speed–relations used for comparison\\
$q_s(0,t)$ & Simplified upstream boundary flow (piecewise definition in text)\\[2pt]

$N_s,\ N_e$ & Cumulative counts used to define average flow windows (veh)\\
$\phi$ & Distribution of the random peak $q_p \sim \phi$\\[4pt]
\multicolumn{2}{@{}l@{}}{\textit{Conventions and qualifiers}}\\
$\bar{\square}$ & Overbar denotes averaging (context-specific)\\
$(\cdot)^{-},\ (\cdot)^{+}$ & Left/right limits or pre/post-peak quantities (defined in context)\\
$(\cdot)^*$ & Counterfactual/auxiliary quantity (defined in context)\\
\bottomrule
\end{tabular}
\endgroup
\caption{Notation.}
\label{tab:notation}
\end{table}

\subsection{The Congested Regime}
First, we consider the case where both vehicles encounter an active queue (congested regime).
\begin{lemma}
	Assume that an active queue exists at the arrival of both vehicles for a peak demand \( q_p^- \), and that \( \Delta\tau_q(q_p^-) = 0 \) holds. Then,  \( \Delta\tau_q(q_p^+) > 0 \).
	\label{lemma:c_final}
\end{lemma}

\begin{proof}[Proof of Lemma \ref{lemma:c_final}]
Under the given assumptions, the flow between the arrival times of A and B at the downstream boundary is \( q_{bn} \), while the difference in cumulative flow between the two vehicles is increasing in \( q_{bn} \). Since the local flow at the downstream boundary between the arrival times remains constant, but the cumulative flow increases, the time gap between the arrivals must increase with \( q_p \).
\end{proof}

\subsection{The Uncongested Regime}
Next, we consider the case where neither vehicle encounters an active queue (uncongested regime). Subsection~\ref{subsec:shocks} establishes that if any two characteristic pairs satisfy the condition 
\[
\Delta \tau_f(q_p^+) > 0,
\] 
then this inequality holds universally, eliminating the need to consider shock waves in uncongested conditions. 
Subsection~\ref{subsec:lb_char} proves that 
\[
\Delta \tau_f(q_p^+) > 0
\] 
when the characteristic corresponding to vehicle~B departs before the vehicle itself, i.e.,
\[
t_{c,B} < t_{pb},
\]
indicating a slower characteristic speed. 
Finally, Subsection~\ref{subsec:lb_tf} extends this result to all remaining cases.

\subsubsection{The Irrelevance of Shock Waves}
\label{subsec:shocks}
Consider two characteristics, \( c_1 \) and \( c_2 \), represented as straight lines in spacetime. We aim to demonstrate that if \( \Delta \tau_f(q_p^+) > 0 \) holds under the assumption of a constant flow \( q \) along \( c_1 \) and \( c_2 \), this condition remains valid even when the flow \( q \) varies due to the occurrence of a shock wave along either characteristic. To prove the statement, we analyze the behavior of the cumulative flow along a characteristic’s trajectory. Assuming the physical correctness of a given characteristic,
and writing $N := N(x,t)$ for brevity, the change in cumulative flow can be described by the following equation:

\[
\frac{dN}{dx} = \frac{dt}{dx} \cdot \frac{\partial N}{\partial t} + \frac{\partial N}{\partial x} = \frac{1}{q'(k)} \cdot q(k) - k = q \cdot k'(q) - k(q).
\]

The following lemma characterizes the monotonic behavior of \( \frac{dN}{dx} \) with respect to \( q \) along the trajectory \( r: t \to x \) of a kinematic wave with transported flow \( q_t \). Since the flow remains constant along \( r \), we have

\[
\frac{dN}{dx} (x,t) = \frac{dN}{dx} (q_t) \quad \text{for all } (t,x) \in r.
\]

\begin{lemma}
\label{lemma:incrconv}
    $\frac{dN}{dx}$ is monotonically increasing and convex in $q$.
\end{lemma}

It is also expedient for our analysis to consider the following relationship between the average flow within an interval bounded by two characteristics at the upstream and downstream ends.

\begin{lemma}
\label{lemma:avgflow}
Let $t_1^d \leq t_2^d$ be arbitrary time points. Define $t_1^a$ as the time when the characteristic line originating from the upstream end at $t_1^d$ reaches the downstream end; similarly, define $t_2^a$ for $t_2^d$. Then:

\begin{enumerate}
    \item If no shock waves occur in the interval $[t_1^d, t_2^a]$ and that $t_2^d \leq t_{pe}$ holds, the following inequality holds:
    \[
        \frac{N(1,t_2^a) - N(1,t_1^a)}{t_2^a - t_1^a} \geq \frac{N(0,t_2^d) - N(0,t_1^d)}{t_2^d - t_1^d}
    \]
    
    \item If $t_1^d \geq t_{pe}$, the following inequality holds:
    \[
        \frac{N(1,t_2^a) - N(1,t_1^a)}{t_2^a - t_1^a} \leq \frac{N(0,t_2^d) - N(0,t_1^d)}{t_2^d - t_1^d}
    \]
\end{enumerate}
\end{lemma}
In the next step, we provide a condition for the relationship between the initial speeds of $A$ and $B$.

\begin{lemma}
\label{lemma:initspeed}
	\( q(0, t_1) \ge q(0, t_2) \), implying that the initial speed of vehicle A is lower than that of vehicle B.
\end{lemma}

Let \( t_{c,A} \) and \( t_{c,B} \) denote the times at which kinematic waves originate from the upstream end and arrive at the downstream end simultaneously with vehicles A and B, respectively. For peak flows of \( q_p^- \) and \( q_p^+ \), the times \( t_{c,A}^- \), \( t_{c,B}^- \), \( t_{c,A}^+ \), and \( t_{c,B}^+ \) are defined analogously. Thus:

\begin{align*}
N(0,t_{c,A}) + \frac{dN}{dx}(q_{c,A}) = N(0,t_1), \\
N(0,t_{c,B}) + \frac{dN}{dx}(q_{c,B}) = N(0,t_1),
\end{align*}

The flow transported by these waves is denoted as \( q_{c,A} \) and \( q_{c,B} \). Therefore:

\begin{align*}
q(0,t_{c,A}) = q_{c,A}, \quad q(0,t_{c,B}) = q_{c,B}
\end{align*}

Again, the definitions of \( q_{c,A}^- \), \( q_{c,B}^- \), \( q_{c,A}^+ \), and \( q_{c,B}^+ \) follow analogously. Relying on the results from lemmas \ref{lemma:avgflow} and \ref{lemma:initspeed}, the statement introduced at the beginning of the subsection can then be proven as follows. An illustration of the argument used in the proof of lemma \ref{lemma:noshocks} can be found in Figure \ref{figure:hatchedareas}.

\begin{lemma}
\label{lemma:noshocks}
Assume
\[
t_{c,A} + k'(q_{c,A}) - t_1 \leq t_{c,B} + k'(q_{c,B}) - t_2.
\]
Then \( \tau_1 \leq \tau_2 \) follows.

\end{lemma}

\begin{figure}
    \centering
    \includegraphics[width=0.5\linewidth]{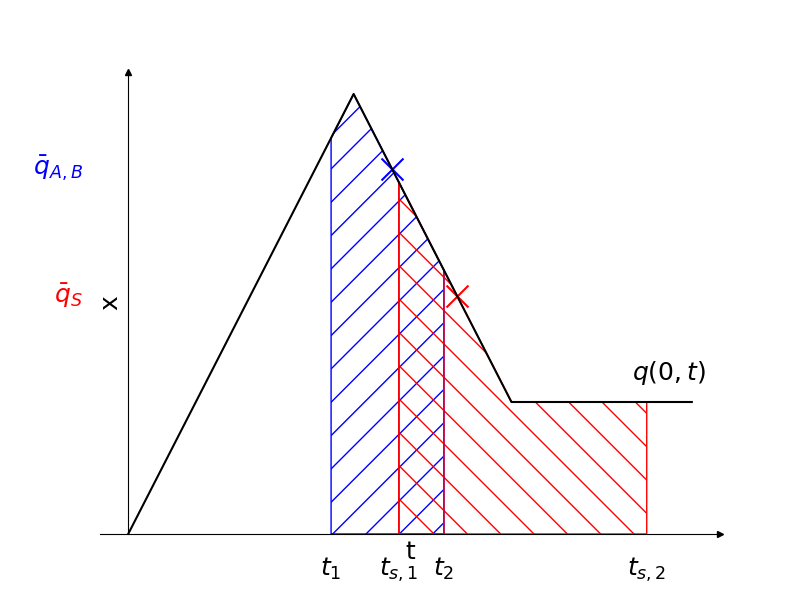}
    \caption{Since the red-hatched area is larger and positioned further to the right than the blue-hatched area, the vertical mean of the blue area (marked with crosses) exceeds that of the red.}
    \label{figure:hatchedareas}
\end{figure}

\subsubsection{A Lower Bound on $t_{c,B}^+$}
\label{subsec:lb_char}
In this subsection, we show the sufficiency of considering cases where \( t_{c,B}^+ \geq t_{pb} \). The proof idea is that a very early departure of the characteristic reaching the downstream boundary at \( t_{B,a}^+ \) - hence a large temporal gap between vehicle departure and its characteristic - arises from the concavity of the fundamental diagram, leading to vehicle speeds well below free-flow. For any instance with \( t_{c,B}^+ < t_{pb} \), there exists a corresponding case with \( t_{c,B}^+ \geq t_{pb} \) that yields an equal or greater \( \Delta \tau_f \). Therefore, we may restrict ourselves to cases with \( t_{c,B}^+ \geq t_{pb} \) without loss of generality, implying that a transformation to satisfy this condition has already been applied where necessary. Denote by $q(x,t, q_p)$ the flow at position x and time t and peak flow realization $q_p$. To prove the result, we begin by assuming \( t_{c,B}^+ \leq t_{pb}\) and \( q(0, t_{c,B}^+, q_p^+) \geq q(0, t_B, q_p^+) \) and construct the adapted instance as described. We first establish the following weaker lower bound for \( t_{c,B}^+ \).

\begin{lemma}
\label{lemma:weaklb}
If \( q(0, t_{c,B}^+, q_p^+) < q(0, t_B, q_p^+) \), then \( \Delta \tau_f(q_p^+) \geq 0 \).
\end{lemma}

Hence, the property \( \Delta \tau_f > 0 \), has already been shown in Lemma \ref{lemma:weaklb} for \( q(0, t_{c,B}^+, q_p^+) < q(0, t_B, q_p^+) \). As discussed in the following lemma, the peak plateau in the trapezoidal boundary flow can be neglected in the analysis of the relevant properties. To achieve this, we consider a suitable instance of the boundary flow with the property $t_{pb}=t_{pe}$, while keeping all other relevant properties unchanged. We then show that the travel time of vehicle B does not increase in this modified instance compared to the original trapezoidal boundary flow. We consider a boundary flow $q(0,t)$ with $q_{pb} < q_{pe}$. Based on this, the simplified version $q_s(0,t)$ is defined as:
\[
q_s(0, t) = 
\begin{cases}
q(0, t), & \text{for } t \leq t_{pb}, \\
q\big(0, t - (t_{pe} - t_{pb})\big), & \text{for } t \geq t_{pe}.
\end{cases}
\]
Furthermore, let \( q_s(x)=0 \) and \( q'_s(x) \) be defined by its derivative as:
\[
q'_s(k) = 
\begin{cases}
q'(k), & \text{for } k \leq q^{-1}\big(q(0, t_{c,B}^+)\big), \\
\text{such that } \frac{dN_s}{dx} q(0, t_{c,B}^+) = \int_{t_{c,B}^+}^{t_B} q_s(0, t) \, dt, & \text{if } k = q_{B,c}^+.
\end{cases}
\]
Then, the following lemma holds:
\begin{lemma}
\label{lemma:peakplateau}
The travel time of vehicle B under the flow-density relationship $q(k)$ and the boundary flow $q(0,t)$ is at least as high as the travel time of vehicle B under the flow-density relationship $q_s(k)$ and the boundary flow $q_s(0,t)$.
\end{lemma}

Following lemma \ref{lemma:peakplateau}, we assume without loss of generality that the boundary flow exhibits a triangular shape without a constant plateau and introduce the simplified notation $t_p := t_{pb} = t_{pe}$.
We proceed by making the following definitions:
\begin{enumerate}

\item \( t^{+*}_{B,c} \) as the time point for which \( q(0,t^{+*}_{B,c}) = q(0,t_{c,B}) \) and \( t^{+*}_{B,c} > t_{c,B} \) holds.
\item \( \bar{q} \) as \( \frac{N(0,t_B)-N(0,t_{c,B}^+)}{t_B-t_{c,B}^+} \) and \( \bar{q}^* \) as \( \frac{N(0,t_B)-N(0,t^{+*}_{B,c}}{t_B-t_{pe}}\).    
\item The modified boundary flow \( \tilde{q}(0,t) \) as:
    \[
    \tilde{q}(0,t) = \begin{cases} 
        q(0,t) & \text{for } t \notin [t_{c,B}^+, t_B] \\
        \bar{q} & \text{for } t \in [t_{c,B}^+, t_B]
    \end{cases}
    \]
\item The boundary flow \( \tilde{q}^*(0,t) \) as:
    \[
    \tilde{q}^*(0,t) = \begin{cases}
        q(0,t) & \text{for } t \notin [t_{c,B}^{+*}, t_B] \\
        \bar{q}^* & \text{for } t \in [t_{c,B}^{+*}, t_B]
    \end{cases}
    \]
    
\item The flow-density relationship \( q^*(k) \) and its associated speed-flow curve \( v^*(q) \) are defined to satisfy
\[
\frac{dN}{dx}\big|_{q(0, t_{c,B})} = \bar{q}^*,
\]
with the conditions \( v^*(\bar{q}) \geq v(\bar{q}^*) \) and \( v^*(\bar{q}) \geq v^*(\bar{q}^*) \), while all other properties of $q^*(k)$ remain arbitrary.
\item The boundary conditions \( \tilde{q}^{\xi,\infty}(0,t) \) for \( \xi \in \{ \cdot, * \} \) as:
\[
\tilde{q}^{\xi,\infty}(0,t) = 
\begin{cases}
\tilde{q}^{\xi}(0,t), & t \leq t_B, \\
\infty, & t > t_B.
\end{cases}
\]
\end{enumerate}

$\bar{q}^* \geq \bar{q}$ holds due to lemma \ref{lemma:weaklb}. An approximation of the properties specified for $q^*(k)$ which is sufficient for the following argument is obtained by:
\begin{itemize}
    \item $q^*(0) = 0$
    \item $q'^*(k) = q'(k)$ for $k$ outside an $\varepsilon$-neighborhood of $q^{-1}(\bar{q}^*)$
    \item $q'^*(k) = (\bar{q}^* - dN/dx(q(\cdot), \bar{q}^*))/\bar{q}^*$ otherwise,
\end{itemize}
where $dN/dx(q(\cdot), \bar{q}^*)$ denotes the function value at $\bar{q}^*$ under the flow-density relationship $q(\cdot)$ and $\varepsilon > 0$ is a sufficiently small real number. It should be noted that the specified construction of \( q^*(k) \) is generally not concave over its entire domain and, therefore, does not represent a valid fundamental diagram according to our own model assumptions or the conventions of traffic flow theory. However, this is irrelevant to the validity of the argument, as \( q^*(k) \) merely serves as an analytical auxiliary construct to demonstrate the desired property of the permissible fundamental diagram \( q(k) \). The uniqueness of the solution \( q(x,t) \), resulting from the definition of \( q^*(k) \) in conjunction with the other model parameters, is sufficient for this purpose. 

Next, we prove that the travel time of vehicle \( B \) under the flow-density relation \( q^* \) establishes an upper bound for its travel time under \( q \). To this end, we first analyze the characteristic curves
\[
s_1 := q'(0,t_{c,B}^+)\cdot(t - t_{c,B}^+) \quad \text{and} \quad s_2 := q'^*(0,t_{c,B}^{+})\cdot(t - t_{c,B}^{+*}).
\]

\begin{lemma}
\label{lemma:char_curves_tt}
The travel time along the characteristic curve $s_1$ is greater than that of $s_2$.
\end{lemma}

We analyze the trajectories of vehicle \( B \) under the following four scenarios:

\begin{enumerate}
    \item[(S$_0$)] Flow-density relation \( q \) with boundary condition \( q(0,t) \)
    \item[(S1)] Flow-density relation \( q \) with boundary condition \( q^\infty(0,t) \)
    \item[(S2)] Modified flow-density relation \( q^* \) with boundary condition \( q^{*,\infty}(0,t) \)
    \item[(S3)] Modified flow-density relation \( q^* \) with boundary condition \( q(0,t) \)
\end{enumerate}

Let \( t_{B,a}^{+,\infty} \) and \( t_{B,a}^{+,\infty,*} \) denote the respective arrival times of \( B \) at the downstream end. We first consider the case in which the characteristic curve \( s_2 \) reaches the downstream end after \( t_{B,a}^{+,\infty,*} \).

\begin{lemma}
\label{lemma:char_curves_tt_2}
    If the characteristic curve \( s_1 \) reaches the downstream boundary at a time later than vehicle \( B \) in S2, then \( s_1 \) reaches the downstream boundary later than \( s_2 \).
\end{lemma}

We now prove that $s_1$ reaches the downstream boundary later than $s_2$ even when arriving before vehicle B in $S2$. To this end, let \( t_{s_i} \) denote the arrival time of characteristic \( s_i \). Let \( t_{B,a,i} \) be vehicle B's arrival time in scenario \( (S_i) \), \( i \in \{0,1,2,3\} \). We first establish the following lemma:

\begin{lemma}
\label{lemma:diffcumflow} 
If
\begin{enumerate}
\item[(i)] \( t_{s_1} \leq t_{B,a,2} \) and
\item[(ii)] \( q_{B,c}^{+} \geq \bar{q} \),
\end{enumerate}
then the difference in cumulative flow at the downstream boundary between times \( t_{B,a,2} \) and \( t_{B,a,1} \) in scenario \( (S_1) \) equals:
\[
\bar{q} \cdot \left(\frac{k(\bar{q}^*)}{\bar{q}^*} - \frac{k(\bar{q})}{\bar{q}}\right).
\]
\end{lemma}

Consider the classification of boundary conditions and resulting wave propagation patterns illustrated in Figure~\ref{figure:trajectories}:
In the upper configuration, where $\bar{q}>q_{B,c}^+$, waves in the second segment of $\bar{q}^\infty(0,t)$ propagate more slowly than in the first segment, generating a rarefaction fan (shown in orange). Vehicle B's trajectory (depicted by black dotted lines) intersects this fan, causing its velocity to increase upon reaching the characteristic associated with flow $\bar{q}$. This represents the only case where the arrival of characteristic $s_i$ (here $s_1$) coincides with vehicle B's arrival.
The middle configuration, where $\bar{q} \leq q_{B,c}^+$, exhibits faster wave propagation in the second segment, resulting in a shock wave. Consequently, the cumulative flow at $s_1$'s arrival is less than $N(0, t_B)$, implying that vehicle B arrives after $s_1$. Since B's trajectory intersects only characteristics carrying flow $\bar{q}$, it maintains a constant velocity $v(\bar{q})$.
The lower configuration, depicting scenario $(S2)$, behaves analogously to the middle case.

\begin{figure}[h!]
    \centering
    \subfloat[S1, $\bar{q}>q_{B,c}^+$]{
        \begin{subfigure}{0.45\textwidth}
            \centering
            \includegraphics[height=5cm]{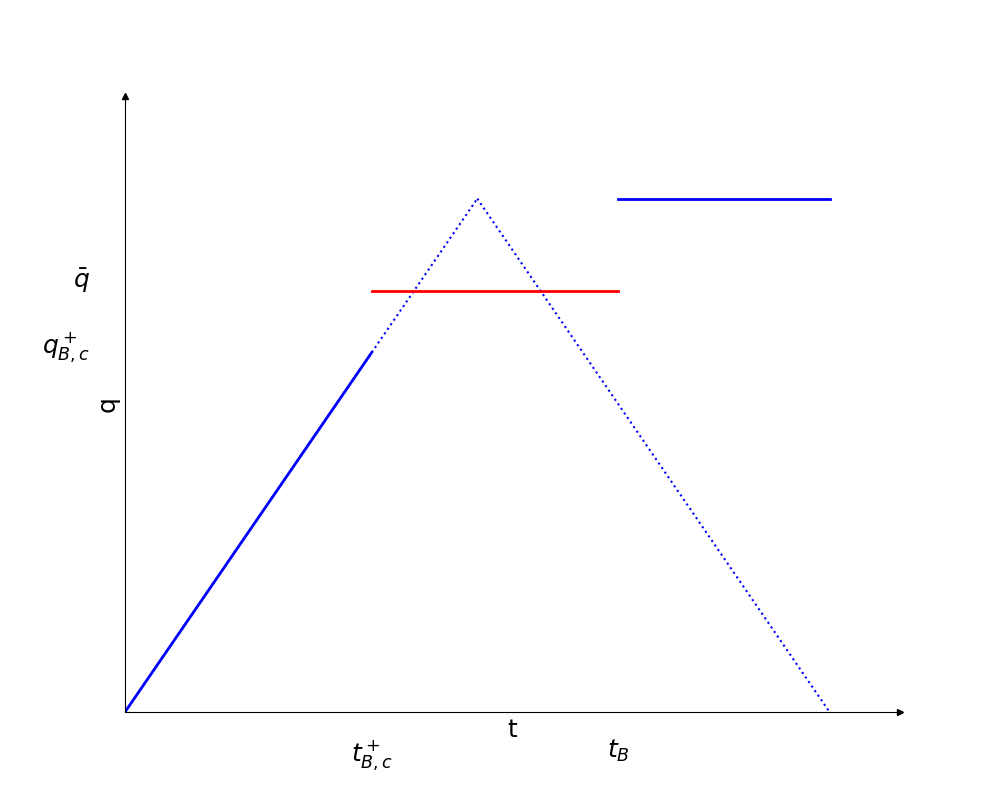}
        \end{subfigure}
        \hfill
        \begin{subfigure}{0.45\textwidth}
            \centering
            \includegraphics[height=5cm]{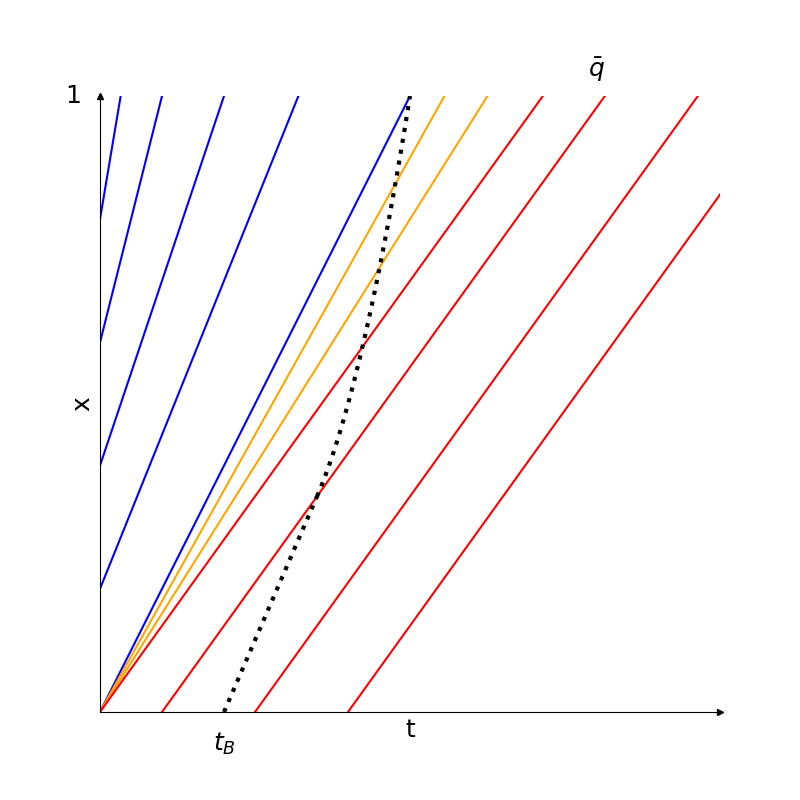}
        \end{subfigure}
    }

    \vskip\baselineskip
    \subfloat[S1, $\bar{q} \leq q_{B,c}^+$]{
        \begin{subfigure}{0.45\textwidth}
            \centering
            \includegraphics[height=5cm]{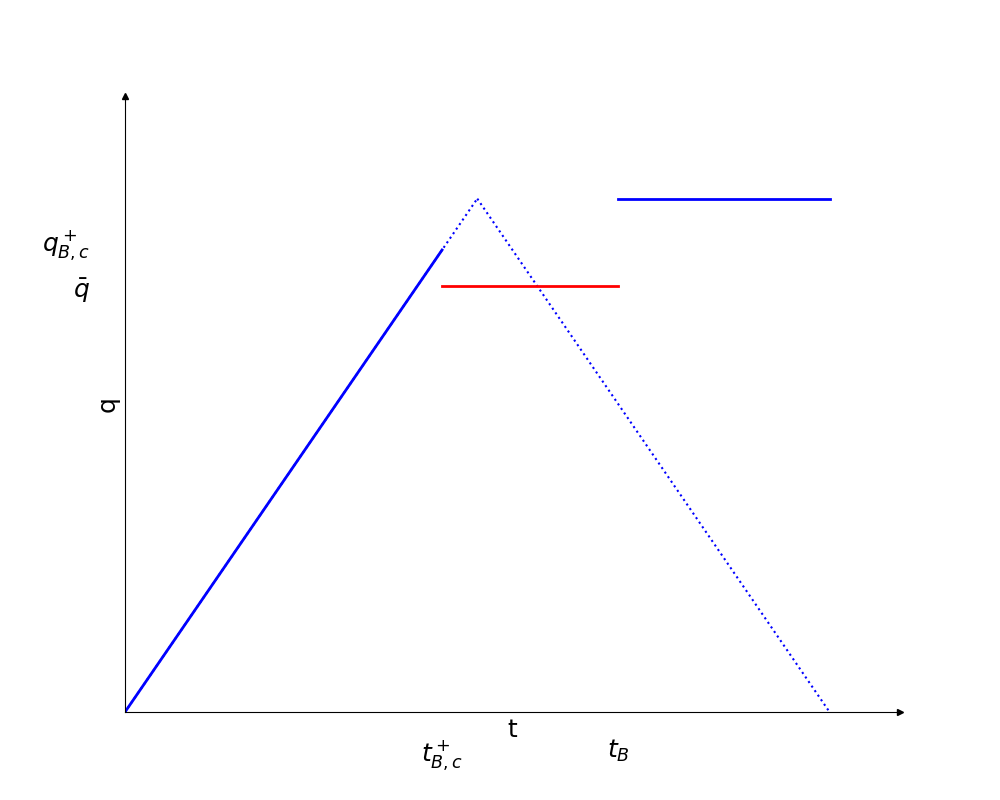}
        \end{subfigure}
        \hfill
        \begin{subfigure}{0.45\textwidth}
            \centering
            \includegraphics[height=5cm]{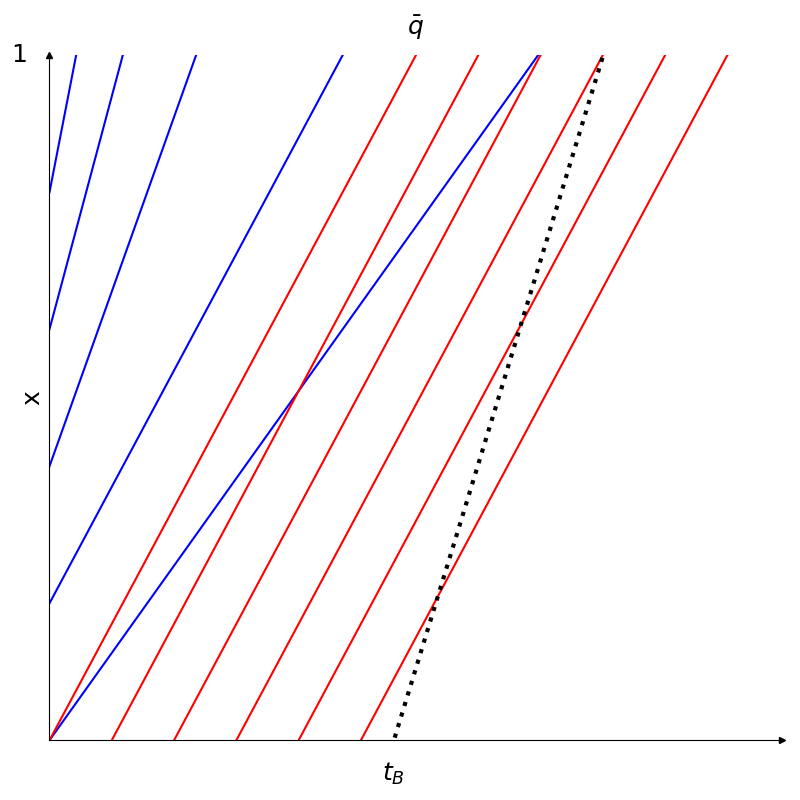}
        \end{subfigure}
    }

    \vskip\baselineskip
    \subfloat[S2]{
        \begin{subfigure}{0.45\textwidth}
            \centering
            \includegraphics[height=5cm]{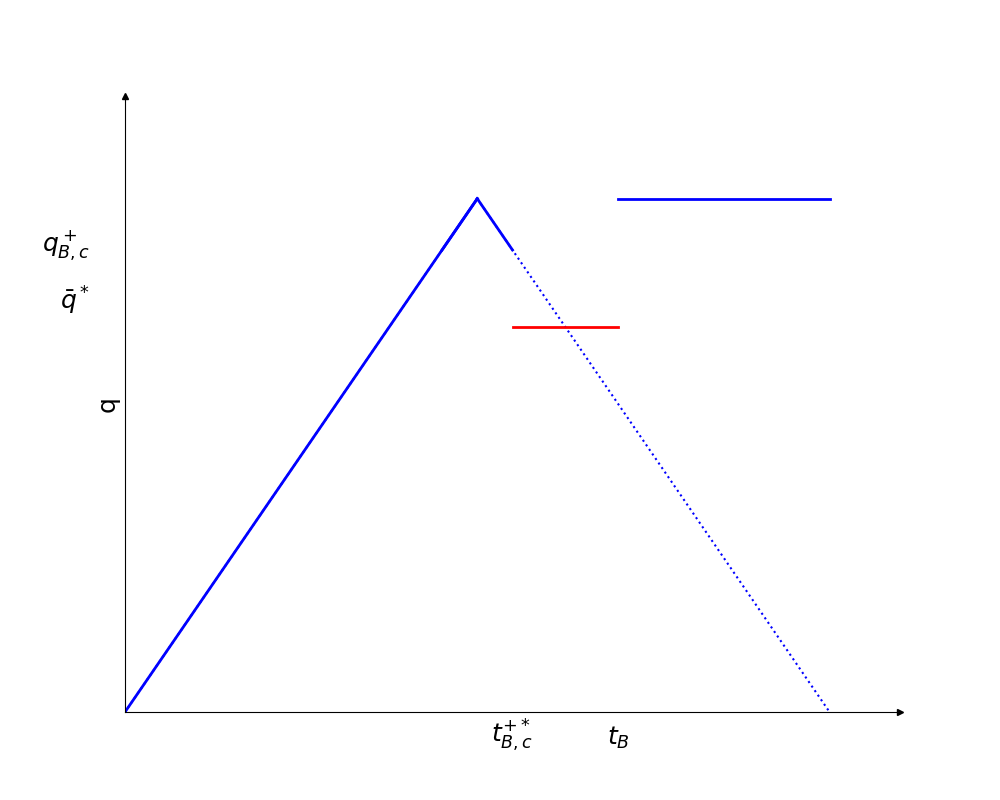}
        \end{subfigure}
        \hfill
        \begin{subfigure}{0.45\textwidth}
            \centering
            \includegraphics[height=5cm]{char_shock.png}
        \end{subfigure}
    }

    \caption{Traffic fLow dynamics under various boundary conditions.}
    \label{figure:trajectories}
\end{figure}

To proceed, we establish bounds for the differential cumulative flow at the downstream boundary between temporal points \( t_{s_1} \) and \( t_{s_2} \) across the respective scenarios. Additionally, we define:
\[
E := \bigg[ k'(q_{B,c}^{+*}) \cdot (q_{B,c}^+ - \bar{q}^*) - k'(q_{B,c}^+) \cdot (q_{B,c}^+ - \bar{q}) \bigg] 
     - \bigg[ k(\bar{q}^*) - k(\bar{q}) \bigg].
\]

For the following steps, denote $N_i(x,t)$ the cumulative flow at position $x$ and time $t$ in scenario $S_i$.
\begin{lemma}
\label{lemma:flowdiffmax}
If $\bar{q} \leq q^+_{B,c}$, $N_1(1, t_{s_2}) - N_1(1, t_{s_1}) \leq E.$ holds in $S_1$.
\end{lemma}

For completing the proof of the statement discussed in this subsection, it is additionally necessary to establish the following upper bound for $E$, initially mentioned in Lemma \ref{lemma:diffcumflow}.

\begin{lemma}
\label{lemma:flowdiffmin}
If $\bar{q} \leq q^+_{B,c}$, $E \leq \bar{q} \cdot \left(\frac{k(\bar{q})}{\bar{q}} - \frac{k(\bar{q}^*)}{\bar{q}^*}\right) $
\end{lemma}

Given the above conditions, we can now prove the statement on characteristic arrival times stated at the beginning of subsection \ref{subsec:lb_char}:

\begin{lemma}
\label{lemma:kw_arr}
Let \( q(0,t) \) be the flow at the upstream boundary of the corridor. Then, the following inequality holds for the characteristic travel times \( k^{*'} \) and \( k' \):
\[
k'^{*}(q(0,t_p)) + t_p \leq k'(q(0,t_{c,B})) + t_{c,B}.
\]
\end{lemma}

The travel time of vehicle B under the flow-density relation $q^*(k)$ provides a lower bound for the travel time under flow-density relation $q(k)$, given the same boundary condition $q(0,t)$. This central insight of this subsection is summarized through the following lemma. 

\begin{lemma}
\label{lemma:lbdeptime}
Let $t_B$ be the departure time of vehicle B and $q(0,t)$ the flow at the upstream boundary. For two flow-density relationships $q(\cdot)$ and $q^*(\cdot)$, the following holds: The travel time of vehicle B under $q(\cdot)$ is greater than or equal to the travel time under $q^*(\cdot)$.
\end{lemma}

Combining Lemmas~\ref{lemma:kw_arr} and~\ref{lemma:lbdeptime} directly establishes that 
\(\Delta \tau_f(q_p^+) > 0\) when \(t_{B,c}^+ \leq t_{p,b}\). 
The non-negativity of \(\Delta \tau_f(q_p^+)\) for larger values of \(t_{B,c}^+\) 
remains to be proven, which we address in the following subsection.

\subsubsection{The Lower Bound on $\Delta \tau_f$}
\label{subsec:lb_tf}
Following the arguments from subsection \ref{subsec:shocks}, we may assume without loss of generality that no shockwaves occur in the interval pertinent to the travel times of A and B. In this subsection, we prove that \(\Delta \tau_f\) is greater than or equal to zero, under the assumption that vehicles A and B do not encounter congestion upon arrival at the downstream end. For this purpose, we normalize the flow \( q(x,t) \) by subtracting a constant background flow \( q_b \). Obviously, the concavity of $ q(k) $ is unaffected by the subtraction of this constant value. This normalization ensures that the system satisfies \( q(0,0) = 0 \), which is useful for the subsequent analysis. Similarly, we define the symbols \( q_{p,r}^- = q_p^- - q_b \) and \( q_{p,r}^+ = q_p^+ - q_b \). In the literature, this method has been applied primarily to the analysis of empirical data (see, e.g., \cite{cassidy1995}), with \cite{Newell1999} being, to our knowledge, the only documented application of this technique to the analytical solution of a traffic flow problem. First, we show that \( \Delta \tau_f(q_{p,r}^+) > \Delta \tau_f(q_{p,r}^-) \) holds and that this implies \( \Delta \tau_f(q_{p}^+) > \Delta \tau_f(q_{p}^-) \). Together with the assumption \( \Delta \tau_f (q_p^-) = 0 \), the originally formulated implication follows. With a slight abuse of notation, we denote the flow of the normalized system again by $q(x,t)$, as well as all symbols derived by $q_r$, as long as the distinction between \( q \) and \( q_r \) is not relevant for the proof.

Let \( t_{c,A}^- \) denote the time at which, under peak flow \( q_p^- \), the characteristic reaching \( x = 1 \) at \( t_{1,A}^- \) exits the upstream corridor boundary. We define \( t_{c,A}^+ \), \( t_{c,B}^- \), and \( t_{c,B}^+ \) analogously. Let $\tilde{t}_{c,A} > t_{c,A}$ be defined such that $q(0,\tilde{t}_{c,A}) = q(0,t_{c,A})$. Denote $q(x,t,q_p)$ the flow at location x in time t, when the peak flow is $q_p$, and let $N(x,t,q_p)$ be defined analogously. For $\tilde{t}_{c,A} > t_{c,A}$ with $q(0,\tilde{t}_{c,A}) = q(0,t_{c,A})$, the arrival time difference of characteristics originating at these points, denoted by \( t_{c,A}^\text{arr} (q_p) \) and \( \tilde{t}_{c,A}^\text{arr} (q_p) \), equals $\tilde{t}_{c,A} - t_{c,A}$ for any peak flow, since:
\begin{align*}
t_{c,A}^\text{arr}(q_p) - \tilde{t}_{c,A}^\text{arr}(q_p) &= \tilde{t}_{c,A} + k'(q(0, \tilde{t}_{c,A}, q_p)) - (t_{c,A} + k'(q(0, t_{c,A}, q_p))) \\
&= \tilde{t}_{c,A} + k'(q(0, t_{c,A}, q_p)) - (t_{c,A} + k'(q(0, t_{c,A},q_p))) \\
&= \tilde{t}_{c,A} - t_{c,A}.
\end{align*}
We prove Theorem \ref{theorem:main} by demonstrating that in the transition from \( q_p^- \) to \( q_p^+ \), vehicle A's arrival time shifts forward relative to \( t_{c,A}^\text{arr} \) , while vehicle B's arrival time shifts backward relative to \( \tilde{t}_{c,A}^\text{arr} \). We first prove this for vehicle A:  

\begin{lemma}
\label{lemma:ASlowerThanChar}
The inequality \( t_{c,A}^{\text{arr}}(q_p^+) \geq t_{1,a}^+ \) holds.
\end{lemma}

We prove the previously mentioned lower bound for \(\Delta \tau_f\) through a logically equivalent statement, using the following additional definitions: Consider two kinematic waves \(s_A^-: t \mapsto x\) and \(s_{A,m}^-: t \mapsto x\) as trajectories in space-time. The first wave \(s_A^-\) leaves the upstream end at time \(t_{c,A}^-\). The second wave \(s_{A,m}^-\) starts at time \(t_{A,m}^-\) during the offset of congestion, where at this time the upstream flow is also \(q_{c,A}^-\). $s_A^+$ and $s_{A,m}^+$ are defined analogously.

We show: During the peak flow transition from \( q_p^- \) to \( q_p^+ \), the time gap between kinematic wave arrivals \((s_A^-\) and \(s_A^+)\) and vehicle A's corresponding arrivals shrinks more rapidly than the time gap between wave arrivals \((s_A,m^-\) and \(s_A,m^+)\) and vehicle B's corresponding arrivals.

To demonstrate the validity of this approach, we first prove the following lemma:

\begin{lemma}
\label{lemma:symmetry_mirror_kw}
It holds that
\[
s_{A,m}^{-1}(1) - s_A^{-1}(1) = s_{A,m}^{+,-1}(1) - s_A^{+,-1}(1)
\]
where \( s_{A, {\cdot, m}}^{-1}: x \mapsto t \) represents the inverse function of \( s_{A, {\cdot, m}} \).
\end{lemma}

For the next step in the proof, we define two new variables:

\begin{itemize}
    \item \( \Delta \Delta \tau_{c} \): This describes the change in the time difference between the arrival of vehicle A and the kinematic wave pair \( s_A^-, s_A^+ \).
    \item \( \Delta \Delta \tau_{c,m} \): This describes the change in the time difference between the arrival of vehicle B and the kinematic wave pair \( s_{A,m}^-, s_{A,m}^+ \).
\end{itemize}

The formal definitions are:

\[
\Delta \Delta \tau_{c} := \left( s_A^{+,-1}(1) - t_A^+ \right) - \left( s_A^{-1}(1) - t_A^- \right),
\]

\[
\Delta \Delta \tau_{c,m} := \left( s_{A,m}^{+,-1}(1) - t_B^- \right) - \left( s_{A,m}^{-1}(1) - t_B^+ \right).
\]

Additionally, we define the corresponding change in the differences in cumulative flows:

\begin{align*}
\Delta \Delta N_{c} &:= 
\big( N(1,s_{A}^{+,-1}(1),q_p^+) - N(0,t_A,q_p^+) \big) \notag \\
&\quad - \big( N(1,s_{A}^{-1}(1),q_p^-) - N(0, t_A, q_p^-) \big) \notag \\
&= N(1,s_{A}^{+,-1}(1),q_p^+) - N(0,t_A,q_p^+),
\end{align*}

\begin{align*}
\Delta \Delta N_{c,m} &:= 
\big( N(1,s_{A,m}^{+,-1}(1),q_p^+) - N(0,t_B,q_p^+) \big) \notag \\
&\quad - \big( N(1,s_{A,m}^{-1}(1),q_p^-) - N(0, t_B, q_p^-) \big).
\end{align*}

The vanishing term in the formula for $\Delta \Delta N_{c}$ results from the fact that the arrival of $s_A$ coincides with the arrival of vehicle $A$ at a peak flow of $q_p^-$. For the change of differences in cumulative flow at the upstream end, we introduce the following notation:

\begin{align*}
\Delta \Delta N_{A,0} &= \left( N(0, t_{c,A}^-, q_p^+) - N(0, t_A, q_p^+) \right) - \left( N(0, t_{c,A}^-, q_p^-) - N(0, t_A, q_p^-) \right), \\
\Delta \Delta N_{B,0} &= \left( N(0, t_{c,A,m}, q_p^+) - N(0, t_B, q_p^+) \right) - \left( N(0, t_{c,A,m}, q_p^-) - N(0, t_B, q_p^-) \right).
\end{align*}

Let $\bar{q}(x;N_s,N_e,q_p)$ denote the average flow at position x measured between two time points: the starting time $t_s$ and ending time $t_e$. At these times, the vehicle counts are $N(x,t_s) = N_s$ and $N(x,t_e) = N_e$ respectively, and the peak flow of this instance is given by $q_p$.

\begin{lemma}
\label{lemma:time_diff_to_char}
The inequality 
\[
\Delta \Delta \tau_{c,m} \geq \Delta \Delta \tau_{c}
\]
holds.
\end{lemma}

\begin{lemma}
\label{lemma:symmetry_aux_inequality}
The inequality
\[
\frac{N(0,t_{c,A,m}^{-},q_p^+)-N(0,t_B^{-},q_p^+)}{\bar{q}(1, N(0,t_{c,A,m}^{-},q_p^+),N(0,t_B,q_p^+),q_p^+)} \leq \frac{N(0,t_{c,A,m}^{-},q_p^-)-N(0,t_B^{-},q_p^-)}{\bar{q}(1, N(0,t_{c,A,m}^{-},q_p^-),N(0,t_B,q_p^-),q_p^-)}
\]
holds.
\end{lemma}

\begin{lemma}
\label{lemma:avgflowaux}
The inequality
\[
    \frac{
    (t_{cB}^- + k'(q(0, t_{cB}, q^-))) - (t_{avg}^- + k'(q(0, t_{\text{avg}^-}, q^-)))
    }{
    (t_{cB}^- + k'(q(0, t_{cB}, q^-))) - (t_{cAm}^- + k'(q(0, t_{cAm}^-, q^-)))
    }
\]
\[
    \leq
\]
\[
    \frac{
    (t_{cB}^+ + k'(q(0, t_{cB}, q^+))) - (t_{\text{avg}}^+ + k'(q(0, t_{\text{avg}}^+, q^+)))
    }{
    (t_{cB}^+ + k'(q(0, t_{cB}, q^+))) - (t_{cAm}^- + k'(q(0, t_{cAm}^-, q^+)))
    }
\]
holds.
\end{lemma}

In the next step, we demonstrate that the statements previously derived for the normalized upstream boundary flow, \( q_r(0,t) = q(0,t) - q_b \), also hold for the original upstream boundary flow \( q(0,t) \). To this end, we define the residual travel time incurred by the de-normalization of the boundary flow,
\begin{equation}
\label{eq:trti_denorm}
\Lambda(q, \tau) = \left[ k'(q + q_b) - k'(q) - \frac{\frac{dN}{dx} (q + q_b) - \frac{dN}{dx} (q) - q_b \cdot (k'(q) - \tau)}{q + q_b} + \tau  \right].
\end{equation}
First, we show the following structural properties of the auxillary function $\Lambda$:
\begin{lemma}
\label{lemma:lambda_incr_pre}
$\Lambda(q)$ is an increasing function in $q$.
\end{lemma}

\begin{lemma}
\label{lemma:lambda_incr}
The inequality $\Lambda (q_A,\tau_A) \leq \Lambda (q_B,\tau_B)$ holds for all $(q_A,\tau_A) \leq (q_B,\tau_B)$.
\end{lemma}

\begin{lemma}
\label{lemma:arrspeeds}
Assume that \( \Delta \tau_f(q_p^+) \geq 0 \) holds for boundary flow \( q_r(0,t) \) and peak flow \( q_p^+ \). Then, the inequality $q_{c,B}^+ \geq q_{c,A}^+$ holds, that is, the arrival speed of vehicle A exceeds that of vehicle B.
\end{lemma}
Next, we deomonstrate the following rule on the acceleration behavior of the respective vehicles under the given conditions:
\begin{lemma}
\label{lemma:accdec}
Assume that \( \Delta \tau (q_p^+) \geq 0 \) holds for boundary flow \( q_r(0,t) \) and peak flow \( q_p^+ \). Denote $v_A(t)$ the speed of vehicle $A$ at time $t$, and $v_B(t)$ the speed of vehicle $B$ at time $t$. Then, the speed of A is increasing along its entire trajectory, while the speed of B is decreasing along its entire trajectory, i.e. $\frac{dv_A}{dt}(t) \geq 0$, $\frac{dv_B}{dt}(t) \leq 0$ for all $t$.
\end{lemma}

In the following lemma, we link the desired property $\tau_f(q_p^+) \geq 0$ from the normalized boundary flow $q_r(0,\cdot)$ back to $q(0,\cdot)$. 
\begin{lemma}
\label{lemma:denorm_pipeline}
Let $\Delta\tau_{f,r}^{+}$ denote the difference in travel times between $A$ and $B$ under upstream boundary flow $q_r(0,t)$, and let $\Delta\tau_f^{+}$ denote the travel time difference under boundary flow $q(0,t)$. Then, the inequality

\[
\Delta\tau_f^{+} \geq \Delta\tau_{f,r}^{+}.
\]
holds.
\end{lemma}
Next, we show that the property $\tau_{f,r}(q_p^+)\geq 0$ assumed in the previous lemma is indeed satisfied.
\begin{lemma}
\label{lemma:pipeline}
Let the upstream boundary flow be given by $q_r(0,t)$. Then, the following inequality holds:
\[
\Delta \tau_f^+ \geq \Delta \tau_f^-.
\]
\end{lemma}
To complete the proof of this subsection's main result, we establish the following auxiliary lemma regarding the average flow between vehicles A and B at their destination:
\begin{lemma}
\label{lemma:flow_compression}
Consider two kinematic waves \( s_i \), \( s_j \) with departure time \( t_i \leq t_p \leq t_j \) from the upstream boundary and associated flows \( q_i \leq q_i \). For any time period \( T \leq s^{-1}(1) \), if \( t_i \geq t_p \), assume additionally that the departure time of the kinematic wave reaching the downstream boundary at \( s^{-1}(1)-T \) exceeds \( t_p \). Then , when increasing the peak flow from \( q_p^- \) to \( q_p^+ \), the difference between the average flow in the interval \([s_i^{-1}(1) - T, s^{-1}(1)]\) and \( q_i \) decreases, while the difference between the average flow in the interval \([s_j^{-1}(1) - T, s_j^{-1}(1)]\) to \( q_j \)  increases. The magnitude of decrease for \( s_i \) is smaller than the magnitude of increase for \( s_j \). 
\end{lemma}

The correctness of the overarching statement follows from Lemmas \ref{lemma:denorm_pipeline} and Lemma \ref{lemma:pipeline}. This completes the proof for the free-flow interval – the period during which both vehicles travel unimpeded by congestion – that vehicle A reaches its destination faster than vehicle B at a peak flow of $q_p^+$.
\begin{lemma}
\label{lemma:f_final}
\( \Delta \tau_f^+ \geq 0 \) holds at a boundary flow of $q(0,t)$.
\end{lemma}

\subsection{The Lower Bound on $\Delta \tau_{fq}$}
Finally, we consider the case when exactly one vehicle encounters a queue.

\begin{lemma}
\label{lemma:queue_order}
If exactly one of the two vehicles A and B encounters a queue along its trajectory, then it is vehicle B.
\end{lemma}

\begin{lemma}
\label{lemma:fc_final}
Assume that \( \Delta \tau_{fq}^- = 0 \), then \( \Delta \tau_{fq}^+ \geq 0 \) follows.
\end{lemma}

This concludes the proof idea outlined at the beginning of the section. For all three separately analyzed regimes - congested (Lemma \ref{lemma:c_final}), uncongested (Lemma \ref{lemma:f_final}), and mixed (Lemma \ref{lemma:fc_final}) - we have shown that the travel time difference $\Delta \tau$ has at most one sign change, specifically from negative to positive. This property satisfies the conditions of Lemma \ref{lemma:characterization}, which establishes that $\operatorname{Var}[\tau_2] \geq \operatorname{Var}[\tau_1]$ holds at every point where $\tau_1$ and $\tau_2$ have equal expectations. Given that the expected travel time is a uni-modal function of departure time and approaches the minimum travel time for both very early and very late departures, this variance relationship directly implies a counterclockwise movement in the $(\operatorname{E}, \operatorname{Var})$ plane. 
\section{Numerical Examples}
\label{section:numerical}

The simulation scenarios using the Cell Transmission Model (CTM) illustrate the theorized hysteresis dynamics. Simulations span 240 time units on a 40-unit corridor with a downstream bottleneck capacity of 25 vehicles per time unit. The fundamental diagram is triangular with a free-flow speed of 1, a critical density \(k_c\) of 60, and a jam density \(k_j\) of 240.

The upstream boundary flow is given by
\( q(0,t) = 20 + \frac{(q_{\text{max}} - 20)}{60} \times t \) for \( t \in [0,60] \), \( q_{\text{max}} \) for \( t \in [60,90] \), \( 10 + q_{\text{max}} - \frac{q_{\text{max}}}{60} \times (t-90) \) for \( t \in [90,150] \), \( 10 \) for \( t \in [150,180] \), and \( 0 \) for \( t \in [180,240] \)
The value of \(q_{\text{max}}\) is normally distributed. 300 simulations were executed across high (\(\mu = 40\)) and low (\(\mu = 30\)) demand scenarios with standard deviations of 10, 15, 20 and 25 of the mean, respectively. Hysteresis is quantified by the area within the loop.

\begin{figure}[h!]
	\centering
	\begin{minipage}{0.8\textwidth}
		\includegraphics[height=0.8\linewidth]{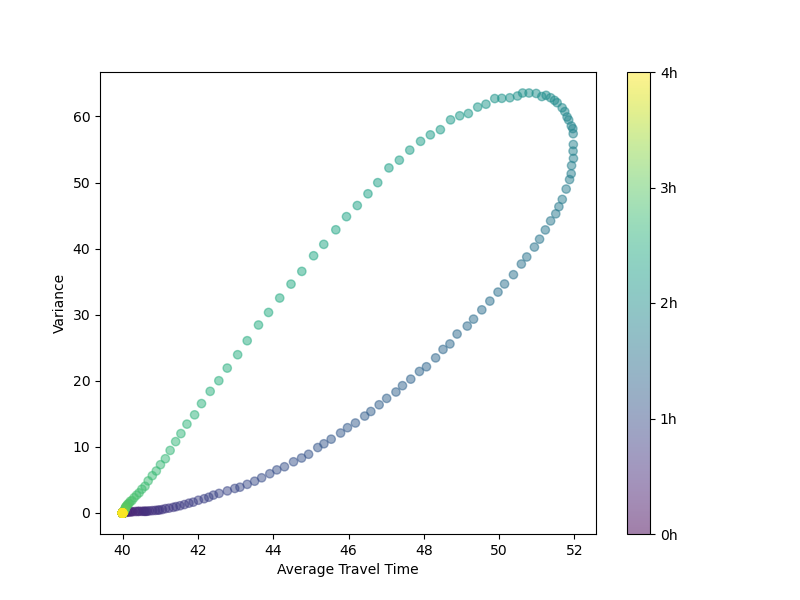} 
		\label{fig:low_demand1}
	\end{minipage}\hfill
	\begin{minipage}{0.8\textwidth}
		\includegraphics[height=0.8\linewidth]{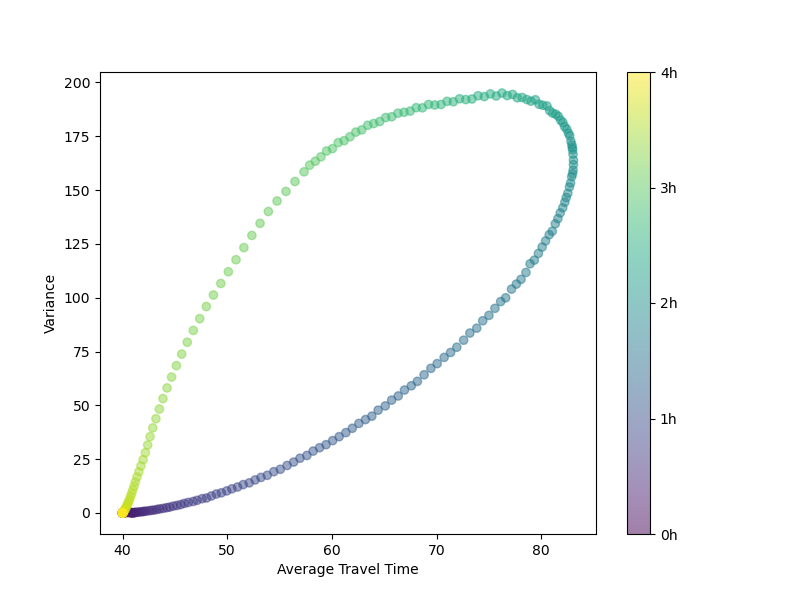} 
		\label{fig:low_demand2}
	\end{minipage}
	\caption{Travel Time And Variances, Low Variance of Boundary Demand}
	\label{fig:low_var_boundary}
\end{figure}

\begin{figure}[h!]
	\centering
	\begin{minipage}{0.8\textwidth}
		\includegraphics[height=0.8\linewidth]{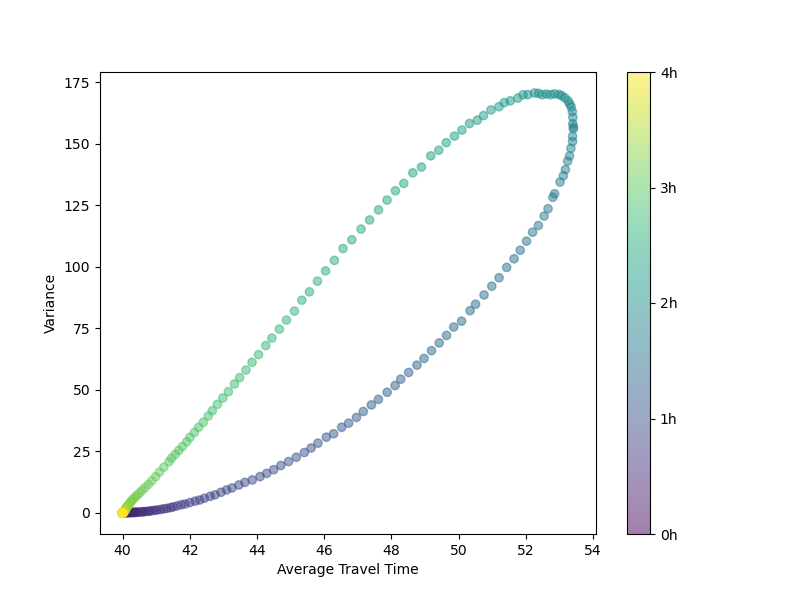} 
		\label{fig:high_demand1}
	\end{minipage}\hfill
	\begin{minipage}{0.8\textwidth}
		\includegraphics[height=0.8\linewidth]{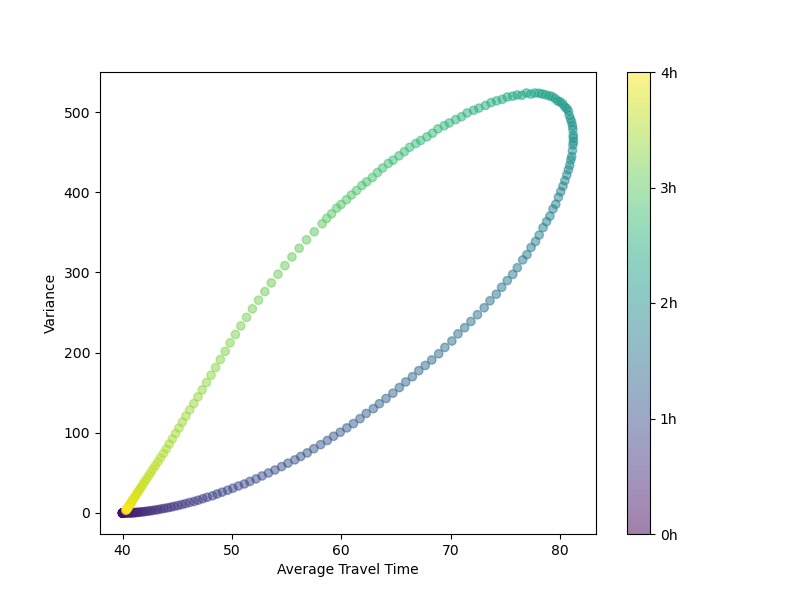} 
		\label{fig:high_demand2}
	\end{minipage}
	\caption{Travel Time And Variances, High Variance of Boundary Demand}
	\label{fig:high_demand}
\end{figure}

Linear regression models, with the standard deviation of \( q_p \) as the independent variable and the area under the hysteresis curve as the dependent variable, exhibit excellent fits for fixed means, yielding slopes \( m_{30} \approx 114.79 \) and \( m_{40} \approx 1088.29 \), and coefficients of determination \( R^2 = 0.9962 \) and \( R^2 = 0.9970 \). These results are supported by Proposition \ref{theorem:main}, highlighting that excess variance between start times \(t_1\) and \(t_2\) reflects sensitivity to changes in the mean rather than variance of \(\phi(q_p)\), as shown by the piecewise linear form of \(\Delta \tau_q\). For the triangular shape of the fundamental diagram, stochasticity of demand has no effect on travel times in uncongested conditions. Graphically, increases in \(\sigma\) do not affect the horizontal mean distance, but mainly increase the vertical variance. Additionally, our analyses indicate that hysteresis magnitude is more influenced by changes in mean peak demand than by its variance.

\section{Empirical Data}
\label{section:empirical}

The empirical analysis is based on measurements conducted on an 8.6 km northbound segment of Interstate 880 in the San Francisco Bay Area, immediately upstream of the bottleneck at the Washington Avenue exit. Data were obtained from the Performance Measurement System (PeMS) of the California Department of Transportation (Caltrans). Figure \ref{figure:freeway_osm} shows the location of the road section under study where data collection was conducted.

\begin{figure}[h!]
  \centering
  \includegraphics[width=0.5\textwidth]{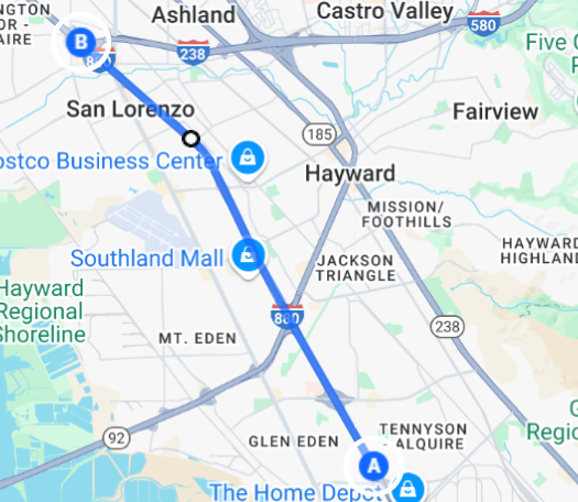}
  \caption{OpenStreetMap view of the studied section of Interstate-880.}
  \label{figure:freeway_osm}
\end{figure}

At this location, the congested exit ramp leads to queue formation across all lanes of I-880, extending approximately one kilometer upstream of the exit. For a detailed description and further treatments of the study site, see \cite{Munoz2002, Hammerl2024, Hammerl2025a, Hammerl2025b}. Travel times were recorded using five loop detectors positioned at approximately equal distances along the corridor. Following the methodology of \cite{Yildirimoglu2015}, we calculated instantaneous travel times in 5-minute intervals between 6:00 AM and 10:00 AM on the 39 U.S. business days from February 5, 2024, to March 29, 2024. From these individual measurements, we computed the mean and empirical variance of travel times for each specific time of day. The results of these calculations are presented in Figure~\ref{figure:mv_emp} as a parametric curve in the $(E, \operatorname{Var})$ plane.

\begin{figure}[htbp]
    \centering
    \includegraphics[width=0.8\textwidth]{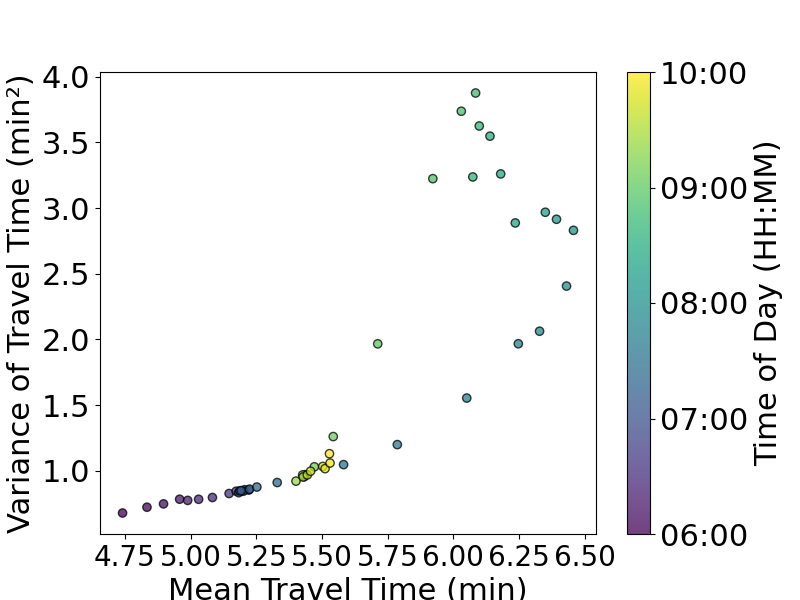}
    \caption{Mean-Variance Relationship of Travel Time During Morning Rush Hour}
    \label{figure:mv_emp}
\end{figure}

The data from I-880 exhibits significant congestion and a clear counterclockwise hysteresis loop in the mean-variance relationship of travel times. No clockwise loops are observed in the data, which may be explained by multiple factors: first, the study site experiences heavy congestion with persistent queuing at the bottleneck, and according to Lemma \ref{lemma:c_final}, queuing-related congestion invariably produces counterclockwise patterns. Additionally, the observation of clockwise loops might require specifically isolating days with conditions conducive to defensive driving behavior, which could be masked in the aggregate analysis of the complete dataset. Zang's work \cite{Zang2017} provides, to our knowledge, the only empirical evidence of clockwise loops in freeway traffic patterns, observing them exclusively during rainy conditions. \cite{Gayah2015} report counterclockwise network-wide traffic patterns, observed both empirically and in microsimulations. In contrast, \cite{Fosgerau2010} identifies clockwise subloops in empirical mean–variance plots, though only in the context of railway operations. Independently, \cite{Akin2011} studied weather effects on the fundamental diagram of traffic flow, demonstrating that rainfall not only decreases the average speed at a given density, but also increases the convex curvature of the speed-density relationship. This convexity corresponds to defensive driving in our framework, which our theoretical analysis shows can lead to clockwise mean-variance loops. Thus, the observation of clockwise loops during rainfall by \cite{Zang2017}, combined with \cite{Akin2011}'s evidence that rain induces defensive driving behavior, provides empirical support for our theoretical framework. Given these theoretical and initial empirical indications, investigating the effects of adverse conditions such as poor weather and low visibility on travel time variability presents a promising direction for future research.

\section{Discussion and Conclusions}
\label{section:conclusion}
This paper examines the time-dependent relationship between the mean and variance of
travel time of vehicular traffic on a single corridor under rush hour like congestion patterns. We show that counterclockwise loops emerge in these conditions under the assumption of speed-density curves that are concave on the non-congested branch.
  
This naturally raises the question of whether this property established for concave speed-density curves extends to convex speed-density curves. However, this is not the case, as demonstrated by the following argument: 
By Lemma \ref{lemma:accdec}, in uncongested conditions, vehicles accelerate during congestion onset and decelerate during congestion offset. Under defensive driving behavior and increasing flow conditions, vehicle A exhibits reduced acceleration, preventing it from matching vehicle B's travel time, as illustrated in Figure \ref{figure:defensive}. This violates the assumption of Lemma \ref{lemma:characterization}, resulting in clockwise sub-loops. 
Figure \ref{figure:defensive} demonstrates this effect through the reduced sensitivity of speed to density changes from \( k_A^- \), \( k_B^- \) to \( k_A^+ \), \( k_B^+ \): as higher flow rates bring the characteristics closer to \( t_1 \) and \( t_2 \), the magnitude of speed changes diminishes.

\begin{figure}[ht]
    \centering
    \begin{tabular}{c c} 
        \includegraphics[width=0.48\textwidth]{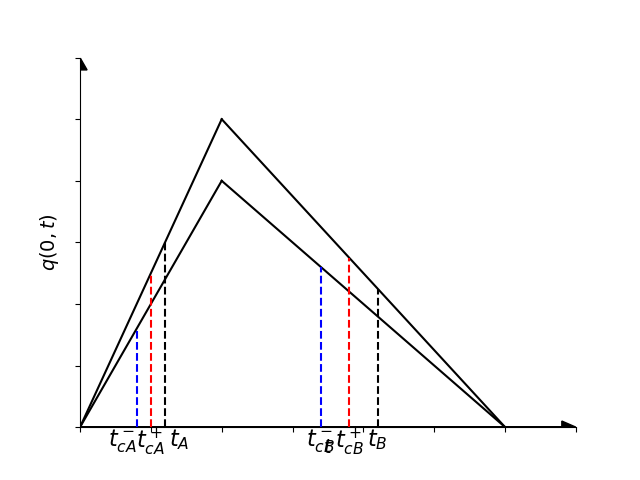} & 
        \includegraphics[width=0.48\textwidth]{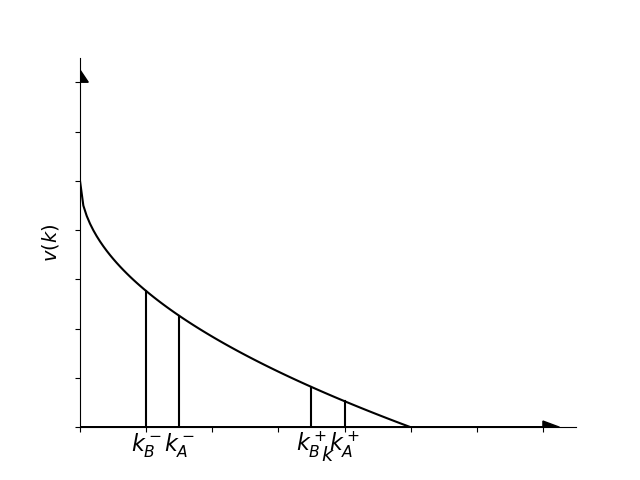} \\
    \end{tabular}
    \caption{Effect of defensive driving on relative delay}
    \label{figure:defensive}
\end{figure}

In our analysis, the calculation of flux changes along the wave trajectory is performed per unit length ($\frac{dN}{dx}$) rather than per unit time ($\frac{dN}{dt}$), as is common in previous variational approaches \cite{Daganzo2005, Daganzo2005b,Daganzo2006}. This spatial perspective is advantageous since the distance traveled by a kinematic wave typically exists as a constant parameter and can often be normalized to 1, while its travel time is a variable quantity. 

In summary, our results provide a rigorous theoretical foundation for empirically observed phenomena in travel time reliability. The analytical relationship between traffic demand and variance propagation offers new insights for travel time prediction and reliability assessment in congested networks. Our findings demonstrate how fundamental properties of the underlying system dynamics manifest in observable reliability metrics, with immediate applications for both offline planning and real-time traffic management. The numerical experiments quantify how system performance metrics scale with input variance, providing practical guidance for robust network design and operations.
Lemma \ref{lemma:characterization} provides a powerful and physically meaningful characterization of systems which exhibit this type of hysteresis. 
\ref{lemma:characterization} implies that in the model used in\cite{Fosgerau2010}, the $(\mathbb{E}, \operatorname{Var})$ curve always forms a single counterclockwise loop. Unlike our model, this result requires no additional assumptions such as aggressive driving behavior. Theorem \ref{theorem:main} demonstrates that the delay difference caused by downstream queues between two units with same expected travel time is indeed a convex function of peak demand. When the boundary flow increases, the upstream flow between $t_1$ and $t_2$ increases linearly, while the bottleneck capacity remains unchanged. For delays caused by increasing density in free-flow conditions, the model must be complemented by additionally assuming driver aggressiveness to meet the fundamental requirement for this pattern to emerge. Interpreting convex speed-density curves as a reflection of timid driving admits a behaviorally and physically conclusive explanation for the occurrence of atypical clockwise loops. 
The extensive discussion on the functional form of this diagram in traffic flow research has resulted in numerous specific proposals. However, the convexity of the speed-density curve has never been explicitly addressed, despite its significant implications for interpreting the underlying driving behavior that we have demonstrated. 
In real traffic networks, travel time delays in uncongested conditions tend to play a minor role compared to delays caused by congestion behind bottlenecks. As demonstrated in lemma \ref{lemma:c_final}, delays in congested intervals are always associated with a counterclockwise loop. Thus, the empirical observations align with our model’s predictions: clockwise loops are rare and, if they occur at all, are accompanied by a counterclockwise loop to the top right, forming a figure-eight double loop.

\appendix

\section*{Appendix}
\label{section:proofs}

\begin{proof}[Proof of Lemma \ref{lemma:incrconv}]
Since $q(k)$ is strictly increasing, it is invertible. We begin by computing 
\[
\frac{d^2 N}{d q \, d x}
 = \frac{d}{dq} \left[ q k'(q) - k(q) \right]
       = k'(q) + q k''(q) - k'(q)
       = q k''(q).
\]
Since $k(q)$ is the inverse of $q(k)$, we have \[
k'(q) = \frac{1}{q'(k)}, \quad k''(q) = -\frac{q''(k)}{[q'(k)]^3}.
\]
Since $ q \geq 0 $ and $ q''(k) \leq 0 $, it follows that $- q q''(k) \geq 0$ and $ [q'(k)]^3 > 0 $. Therefore, $\frac{d^2 N}{d q \, d x}
 \geq 0$, so $\frac{dN}{dx}$ is monotonically increasing in $q$.

We compute \[
    \frac{d^3 N}{d q^2 \, d x} (q) = \frac{d}{dq} \left[\frac{d^2 N}{d q \, d x} \right] = \frac{d}{dq} \left( \frac{-q \cdot q''(k)}{[q'(k)]^3} \right)
    \]
using the quotient and product rules. Let 
\[
N = q q''(k), \quad D = [q'(k)]^3,
\]
then
\[
\frac{d^3 N}{d q^2 \, d x} (q)= -\left( \frac{N'}{D} - \frac{ND'}{D^2} \right). 
\]
We can express derivatives with respect to $q$ in terms of derivatives with respect to $k$: $\frac{d}{dq} = \frac{1}{k'(q)} \frac{d}{dk} = q'(k) \frac{d}{dk}.$ 
Compute N' and D':
\[
N' = q'(k) \frac{d}{dk} [q q''(k)] = q'(k)[q'(k)q''(k) + q q'''(k)],
\]
\[
D' = q'(k) \frac{d}{dk} [q'(k)^3] = q'(k)[3[q'(k)]^2q''(k)] = 3[q'(k)]^3q''(k).
\]
Substitute back into $\frac{d^3 N}{d q^2 \, d x}$:
\[
\frac{d^3 N}{d q^2 \, d x} (q) = - \left( \frac{q'(k)[q'(k)q''(k) + q q'''(k)]}{[q'(k)]^3} - \frac{q q''(k) \cdot 3[q'(k)]^3q''(k)}{[q'(k)]^6} \right).
\]
Simplify:
\[
\frac{d^3 N}{d q^2 \, d x}(q) = - \left( \frac{q'(k) q''(k) + q q'''(k)}{[q'(k)]^2} - \frac{3q q''(k)}{[q'(k)]^5} \right).
\]
Since $q \geq 0, q'(k) > 0, q''(k) \leq 0$, and  $q'''(k) \leq 0$, we have:
\begin{enumerate}
    \item $- \frac{q'(k) q''(k)}{[q'(k)]^2} \geq 0,$
    \item $- \frac{q q'''(k)}{[q'(k)]^2} \geq 0,$
    \item $\left( \frac{-3q q''(k)}{[q'(k)]^5} \right)\geq 0.$
\end{enumerate}
Therefore,  $\frac{d^3 N}{d q^2 \, d x} (q) \geq 0$, so $\frac{dN}{dx} (q)$ is convex in q.
\end{proof}

\begin{proof}[Proof of Lemma \ref{lemma:avgflow}]
We restrict the proof to the first statement, as the second part follows by analogous arguments with reversed inequalities. 

Let \( r \in [0, 1] \) and define $t_0 = r \cdot t_1^d + (1 - r) \cdot t_2^d, \quad t_2 = r \cdot t_1^a + (1 - r) \cdot t_2^a.$

Consider the characteristic line starting at \( t_0 \). It reaches \( x = 1 \) at time \( t_1 = t_0 + k'(t_0) \). Due to the convexity of \( k' \), we have $ k'(t_0) \leq r \cdot k'(t_1^d) + (1-r) \cdot k'(t_2^d) $.
Furthermore, $t_2 = r(t_1^d + k'(t_1^d)) + (1 - r)(t_2^d + k'(t_2^d))= t_0 + rk'(t_1^d) + (1 - r)k'(t_2^d)$. Combining these results yields $ t_1 = t_0 + k'(t_0) \leq t_0 + rk'(t_1^d) + (1 - r)k'(t_2^d) = t_2 $. Thus, the characteristic line from \( (0,t_0) \) reaches \( x = 1 \) no later than \( t_2 \), implying $ q(1, r \cdot t_1^a + (1-r) \cdot t_2^a) \geq q(0, r \cdot t_1^d + (1-r) \cdot t_2^d) $.

Next, we introduce a parameterization of the relevant intervals by defining \( u \in [0,1] \) and $ h(u) = t_1^d + u(t_2^d - t_1^d), \quad i(u) = t_1^a + u(t_2^a - t_1^a) $. Since \( q(0,h(u)) \leq q(1,i(u)) \) holds for all \( u \), integration over \( [0,1] \) yields
$ \int_0^1 q(0,h(u)) du \leq \int_0^1 q(1,i(u)) du $. Using the substitutions \( dh = (t_2^d-t_1^d)du \) and \( di = (t_2^a-t_1^a)du \), we obtain:
  \[
        \int_{t_1^d}^{t_2^d} \frac{q(0,t)}{t_2^d - t_1^d}dt \geq \int_{t_1^a}^{t_2^a} \frac{q(1,t)}{t_2^a - t_1^a}dt \implies \frac{1}{t_2^d - t1^d} \int_{t_1^d}^{t_2^d} q(0,t) dt \geq \frac{1}{t_2^a - t1^a} \int_{t_1^a}^{t_2^a} q(1,t) dt.\]

This final inequality demonstrates that the mean flow in \([t_1^d, t_2^d]\) at the upstream end is lower than that in \([t_1^a, t_2^a]\) at the downstream end, which implies the inequality stated in the lemma.
\end{proof}

\begin{proof}[Proof of Lemma \ref{lemma:initspeed}]
Assume \( q(0,t_1) < q(0,t_2) \) for a contradiction. Then, \( t_1 \leq t_{pb} \). Let \( t_{c1} \) and \( t_{c2} \) be the times at which characteristic lines through the intersections of A and B's trajectories with the queue tail leave the upstream boundary. Then, \( t_{c1} < t_{c2} \) and \( q(t_{c1}) < q(t_{c2}) \) and must hold. We distinguish two cases:
	
	\begin{enumerate}
		\item If \( q_{c2} > q(0,t_1) \), then \( v(0,t_1) \geq v(0,t_{c2}) \) and \( v(0,t_{c2}) \leq \bar{v}_B \), where \( \bar{v}_B \) is B's average speed upstream of the bottleneck, while \( v(0,t_1) \geq \bar{v}_A \) implies \( \bar{v}_A > \bar{v}_B \). Since \( q(0,t_{c1}) < q(0,t_{c2}) \), the queue length increases by B's entry than by A's. Thus, \( \tau_{bn}(t_1) \leq \tau_{bn}(t_2) \), and \( \tau(t_1) < \tau(t_2) \).
		
		\item \( q_{c2} \leq q(0,t_1) \). If \( q_{c2} \leq q(0,t_1) \), and there is no active queue when A arrives, the average flow between A and B at the upstream boundary is at least \( q(0,t_1) \) and at the downstream boundary at most \( q(0,t_{c2}) \). The difference in cumulative flow remains constant, hence $\tau(t_1)>\tau(t_2)$. If an active queue exists for A, it also exists for B, and the flow at the downstream boundary between arrivals is \( q_{bn} \), while it is higher than \( q_{bn} \) between departures upstream. Hence, \( \tau(t_1) < \tau(t_2) \).
	\end{enumerate}
	
In summary, in both cases, \( \tau(t_1) < \tau(t_2) \), and hence \( \mathbb{E}[\Delta \tau] < 0 \), a contradiction to the assumptions.
\end{proof}

\begin{proof}[Proof of Lemma \ref{lemma:noshocks}]
First, we observe that the lemma holds by construction if no shock waves occur in the interval \([t_{c,A}, t_{c,B} + k'(q_{c,B})]\). However, if a shock wave does occur, we consider the earliest time \(t_s\) of its occurrence and distinguish between the following cases:
\begin{enumerate}
    \item \(t_s \leq t_{c,A} + k'(q_{c,A})\). In this case, Lemma \ref{lemma:shockpriority} implies that the characteristic originating from \(t_{c,A}\) is intersected by a later starting, faster characteristic \(q_{c,A}\). Let \(q_{s,1}\) be the flow at the downstream boundary at time \(t_{1,a}^+\), and let \(t_{s,1}\) be the time at which the characteristic carrying this flow emanates from the upstream end of the corridor. Furthermore, let \(\Delta N\) be the difference in cumulative flow between vehicles \(A\) and \(B\), that is, \(\Delta N = N(0, t_1) - N(0, t_2)\). The cumulative flow at the downstream boundary between \(t_{1,a}^+\) and \(t_{2,a}^+\) must also equal \(\Delta N\). Denoting the flow at the downstream boundary at \(t_{2,a}^+\) as \(q_{s,2}\) and the starting time of its corresponding characteristic as \(t_{s,2}\), then for the piecewise linear path in the \(x\)-\(t\)-plane
    \[
    C := (0, t_{s,1}) \rightarrow (1, t_{1,a}^+) \rightarrow (1, t_{2,a}^+) \rightarrow (0, t_{s,2}) \rightarrow (0, t_{s,1})
    \]
    we have by definition \(\oint_C q \, dt - k \, dx = 0\). Since the upstream boundary flow decreases monotonically after $t_s$,  \(q_{s,1} \geq q_{s,2}\) holds and thus \(dN/dx(q_{s,1}) \geq dN/dx(q_{s,2})\). It follows that
    \[
    N(0, t_{s,1}) - N(0, t_{s,2}) \geq N(1, t_{2,a}^+) - N(1, t_{1,a}^+) = \Delta N.
    \]

    Since \(t_{s,1}\) necessarily lies in the interval of decreasing flow at the upstream boundary, we can conclude that
    \[
    \frac{N(0, t_2) - N(0, t_1)}{t_2 - t_1} \geq \frac{N(0, t_{s,2}) - N(0, t_{s,1})}{t_{s,2} - t_{s,1}}.
    \]
    An illustration of this argument can be found in Figure \ref{figure:hatchedareas}. According to Lemma \ref{lemma:avgflow}, we also have
    \[
    \frac{N(0, t_{s,2}) - N(0, t_{s,1})}{t_{s,2} - t_{s,1}} \geq \frac{N(1, t_{a,2}^+) - N(1, t_{a,1}^+)}{t_{a,2}^+ - t_{a,1}^+}.
    \]
    This yields
    \[
    \frac{N(0, t_2) - N(0, t_1)}{t_2 - t_1} \geq \frac{N(1, t_{a,2}^+) - N(1, t_{a,1}^+)}{t_{a,2}^+ - t_{a,1}^+}
    \]
    and \(N(0, t_2) - N(0, t_1) = N(1, t_{a,2}^+) - N(1, t_{a,1}^+)\), from which it follows that \(t_{a,2}^+ - t_{a,1}^+ \geq t_2 - t_1\). This confirms the statement of the lemma.
    \item \(t_s > t_{c,A} + k'(q_{c,A})\). In this case, the flow remains constant along the characteristic originating at \(t_{c,A}\), and \(t_{1,a}^+ = t_{c,A} + k'(q_{c,A})\). Following \ref{lemma:newell}, the physically relevant characteristic is determined by the minimum of cumulative flows. Hence, \(N(1, t_{c,B} + k'(q_{c,B})) < N(0, t_2)\), which implies \(\tau(t_2, q_p^+) > \tau(t_1, q_p^+)\).
\end{enumerate}
\end{proof}

\begin{proof}[Proof of Lemma \ref{lemma:weaklb}]
According to Lemma \ref{lemma:initspeed}, it must hold that \( q(0, t_A, q_p^+) \geq q(0, t_B, q_p^+) \). Therefore, the average flow at the upstream end over the interval \([t_A, t_B]\) satisfies \( \overline{q}_{[t_A, t_B]}^{\text{u}} \geq q(0, t_B, q_p^+) \). By construction, we have \( t_{c,B}^+ \leq t_{pb} \); since vehicle \( A \) arrives before vehicle \( B \), it also follows that \( t_{c,A}^+ < t_{c,B}^+ \). The starting points of both characteristics are thus within the interval of increasing boundary flow, so the flow transported by the characteristics reaching the downstream end between the arrival of vehicles \( A \) and \( B \) is at most \( q(0, t_{c,B}^+, q_p^+) \). Consequently, the average downstream flow over the interval \([t_{A,d}^+, t_{B,d}^+]\) satisfies \( \overline{q}_{[t_{A,d}^+, t_{B,d}^+]}^{\text{d}} \leq q(0, t_{c,B}^+, q_p^+) < q(0, t_B, q_p^+) \leq \overline{q}_{[t_A, t_B]}^{\text{u}} \). Since the average flow between the two vehicles is higher at the upstream than at the downstream end, while the difference in cumulative flow between them remains constant, it follows that \( t_{B,d}^+ - t_{A,d}^+ > t_B - t_A \) and hence \( \tau_B - \tau_A = \Delta \tau_f > 0 \).
\end{proof}

\begin{proof}[Proof of Lemma \ref{lemma:peakplateau}]
\end{proof}

\begin{proof}[Proof of Lemma \ref{lemma:char_curves_tt}]
Due to the normalization of the corridor length, the travel time of \( s_1 \) is \( k'(0, t_{c,B}^+) \), while that of \( s_2 \) is \( k'^*(0, t_{c,B}^{+}) \). We have:
\begin{align*}
N(0, t_B) - N(0, t_{c,B}^+) &> N(0, t_B) - N(0, t_{c,B}^{*+}) \\
&\implies \bar{q} > \bar{q}^* \\
&\implies \frac{dN}{dx}(q(\cdot), q_{B,c}^+) > \frac{dN}{dx}(q^*(\cdot), q_{B,c}^{+*}) \\
&\implies q_{B,c} \cdot k'(q_{B,c}) - k(q_{B,c}) > q_{B,c} \cdot k'^*(q_{B,c}) - k^*(q_{B,c}).
\end{align*}
Given that $k^{'*}(\cdot)$ and $k'(\cdot)$ are equal except for a marginally small interval where they differ by a bounded amount, we can replace $k^{*}(q_{B,c})$ with $k(q_{B,c})$. This yields:
\begin{align*}
q_{B,c} \cdot k'(q_{B,c}) - k(q_{B,c}) &> q_{B,c} \cdot k'^*(q_{B,c}) - k(q_{B,c}) \\
&\implies k'(q_{B,c}) > k'^*(q_{B,c}),
\end{align*}
which proves the lemma.
\end{proof}

\begin{proof}[Proof of Lemma \ref{lemma:char_curves_tt_2}]
The path integral along the trajectory of \( s_2 \) satisfies
\[
N(0, t_{c,B}^{+*}) + \int_{s_2} q(0,t_{c,B}^{+*}) \, dt - k(0,t_{c,B}^{+*}) \, dx = N(0, t_B).
\]
According to \ref{lemma:newell}, the characteristic curve that satisfies the integral conservation law \ref{eq:conservation} is the one associated with the lowest cumulative flow. Therefore, the arrival time of \( s_2 \) provides a lower bound for the arrival time of vehicle \( B \). Since the arrival time of vehicle \( B \) itself is a lower bound for the arrival time of \( s_1 \), the lemma follows.
\end{proof}

\begin{proof}[Proof of Lemma \ref{lemma:diffcumflow}]
Consider the trajectories of vehicle B in both scenarios. Since the boundary flow becomes infinite after time \( t_B \), the characteristics originating from the upstream boundary after this time propagate very slowly. Thus, vehicle B's trajectory cannot be intersected by shock waves in either scenario. Moreover, in scenario \( i \), no wave starting before \( t_{c,B}^{+,i} \) can intersect B's trajectory, as the characteristic \( s_i \) starts at this time and provides a lower bound for B's arrival time due to its associated cumulative flow.

As a consequence, B travels through homogeneous traffic with flow \( \bar{q} \) in \( (S_1) \) and \( \bar{q}^* \) in \( (S_2) \). The interaction between vehicle B's trajectory and traffic wave propagation is illustrated in figure \ref{figure:trajectories}. The respective travel times are \( \frac{k(\bar{q})}{\bar{q}} \) and \( \frac{k(\bar{q}^*)}{\bar{q}^*} \), yielding a temporal difference of
\[
\frac{k(\bar{q})}{\bar{q}} - \frac{k(\bar{q}^*)}{\bar{q}^*}.
\]
Given our assumption \( t_{s_1} \leq t_{B,2} \), the upstream boundary flow in \( (S_1) \) equals \( \bar{q} \) over the interval \( [t_{B,2}, t_{B,1}] \). The difference in cumulative flow stated in the lemma follows from multiplying this flow rate by the temporal difference.
\end{proof}

\begin{proof}[Proof of Lemma \ref{lemma:flowdiffmax}]
Since $N_1(0, t_B)=N_1(0, t_B)$, we denote $N_1(0, t_B)=N_1(0, t_B)=N(0,t_B)$.
In figure \ref{figure:trajectories} (middle image pair), the trajectory of \( s_1 \) is depicted by the rightmost blue line in the right picture. As indicated by the red lines overlapping the trajectory and formally shown in lemma \ref{lemma:shockpriority}, the local flow \( q(x,t) \) along this trajectory  remains constant at \( q_{B,c}^+ \). Thus, \( N_1(0, t_B) - N_1(1, t_{s_1}) \) can be expressed as the difference of two path integrals:

\[
\begin{aligned}
N_1(0, t_B) - N_1(1, t_{s_1}) = & \int_{s_1} \left( q_{B,c}^+ \, dt + k(q_{B,c}^+) \, dx \right) \\
& - \int_{s_1} \left( \bar{q} \, dt + k(\bar{q}) \, dx \right) \\
= & \, (q_{B,c}^+ - \bar{q}) \, k'(q_{B,c}^+) - \left( k(q_{B,c}^+) - k(\bar{q}) \right).
\end{aligned}
\]

The characteristic curve \( s_2 \) leaves the upstream boundary later than \( s_1 \). Due to this temporal sequence, we conclude geometrically that the local flow \( q(x,t) \) must also maintain the constant value \( q_{B,c}^+ \) along the entire characteristic curve \( s_2 \). Consequently, the quantity \( N(1, t_{s_2}) \) can be determined through the following path integral:

\[
\begin{aligned}
N_1(0, t_B) - N_1(1, t_{s_2}) = & \int_{s_2} \big( q_{B,c}^{+*} \, dt + k(q_{B,c}^{+*}) \, dx \big) 
- \int_{s_2} \big( \bar{q}^* \, dt + k(q_{B,c}^+) \, dx \big) \\
= & \big( q_{B,c}^{+*} \cdot k'(q_{B,c}^{+*}) - k(q_{B,c}^{+*}) \big) 
- \big( \bar{q}^* \cdot k'(q_{B,c}^{+*}) - k(\bar{q}^*) \big) \\
= & \big( q_{B,c}^+ \cdot k'(q_{B,c}^{+*}) - k(q_{B,c}^+) \big) 
- \big( \bar{q}^* \cdot k'(q_{B,c}^{+*}) - k(\bar{q}^*) \big) \\
= & k'(q_{B,c}^{+*}) \cdot (q_{B,c}^+ - \bar{q}^*) - (k(q_{B,c}^+) - k(\bar{q}^*)).
\end{aligned}
\]

We subtract these quantities to verify the inequality stated in the lemma:

\[
\begin{aligned}
N_1(1, t_{s_2}) - N_1(1, t_{s_1}) = & \, (N(0, t_B) - N_1(1, t_{s_1})) - (N(0, t_B) - N_1(1, t_{s_2})) \\
= & \bigg[ k'(q_{B,c}^{+*}) \cdot (q_{B,c}^+ - \bar{q}^*) - \big( k(q_{B,c}^+) - k(\bar{q}^*) \big) \bigg] \\
& - \bigg[ (q_{B,c}^+ - \bar{q}) \cdot k'(q_{B,c}^+) - \big( k(q_{B,c}^+) - k(\bar{q}) \big) \bigg] \\
= & \bigg[ k'(q_{B,c}^{+*}) \cdot (q_{B,c}^+ - \bar{q}^*) - k'(q_{B,c}^+) \cdot (q_{B,c}^+ - \bar{q}) \bigg] \\
& - \bigg[ \big( k(q_{B,c}^+) - k(\bar{q}^*) \big) - \big( k(q_{B,c}^+) - k(\bar{q}) \big) \bigg] \\
= & \bigg[ k'(q_{B,c}^{+*}) \cdot (q_{B,c}^+ - \bar{q}^*) - k'(q_{B,c}^+) \cdot (q_{B,c}^+ - \bar{q}) \bigg] \\
& - \bigg[ k(\bar{q}^*) - k(\bar{q}) \bigg].
\end{aligned}
\]
\end{proof}

\begin{proof}[Proof of Lemma \ref{lemma:flowdiffmin}]
For a convex function \( k(q) \), the Legendre transform \( k^*(p) \) is defined as
\[
k^*(p) = \sup_q \{ pq - k(q) \}.
\]
By the Fenchel-Young equality, 
\[
k(q) + k^*(p) = pq
\]
holds for all \( k(q) \) and \( k^*(p) \), where \( p = k'(q) \) and \( q = k^*(p) \). We set $k(q_{B,c}^{+*}) = p^{+*} q_{B,c}^{+*} - k^*(p^{+*})$ and $k(q_{B,c}^+) = p^{+} q_{B,c}^+ - k^*(p^{+})$. 

We rewrite $E$ as follows: 
\[
\begin{aligned}
E = & \bigg[ k'(q_{B,c}^{+*}) \cdot (q_{B,c}^+ - \bar{q}^*) 
        - k'(q_{B,c}^+) \cdot (q_{B,c}^+ - \bar{q}) \bigg] 
        - \bigg[ k(\bar{q}^*) - k(\bar{q}) \bigg] \\
= & \bigg[ p^{+*}(q_{B,c}^{+} - \bar{q}^*) 
        - p^+(q_{B,c}^+ - \bar{q}) \bigg] 
        - \bigg[ k(\bar{q}^*) - k(\bar{q}) \bigg] \\
= & \bigg[ p^{+*}q_{B,c}^{+} - p^{+*}\bar{q}^* 
        - p^+q_{B,c}^{+} + p^+\bar{q} \bigg] 
        - \bigg[ k(\bar{q}^*) - k(\bar{q}) \bigg] \\
= & \bigg[ (p^{+*} - p^+)q_{B,c}^{+} 
        + p^+\bar{q} - p^{+*}\bar{q}^* \bigg] 
        - \bigg[ k(\bar{q}^*) - k(\bar{q}) \bigg] \\
= & (p^{+*} - p^+)(q_{B,c}^{+} - \bar{q}^*) 
      + p^+(\bar{q} - \bar{q}^*) 
      - \big[k(\bar{q}^*) - k(\bar{q})\big].
\end{aligned}
\]

Note that \(\big(\bar{q}^*\big)(p^+ - p^{+*}) + \big(p^{+*} - p^+\big)\bar{q}^* = 0\). Since \(p^{+*} \leq p^+\) and \(q_{B,c}^{+} \geq \bar{q}^*\), the term \((p^{+*} - p^+)(q_{B,c}^{+} - \bar{q}^*) \leq 0\).

Thus:
\[
E \leq p^+(\bar{q} - \bar{q}^*) - \big[k(\bar{q}^*) - k(\bar{q})\big].
\]

From the convexity of \(k(q)\):
\[
k(\bar{q}) - k(\bar{q}^*) \geq k'(\bar{q}^*)(\bar{q} - \bar{q}^*).
\]

Therefore:
\[
-[k(\bar{q}^*) - k(\bar{q})] \leq -k'(\bar{q}^*)(\bar{q} - \bar{q}^*) = k'(\bar{q}^*)(\bar{q}^* - \bar{q}).
\]

Since \(\bar{q}^* \leq \bar{q}, \bar{q}^* - \bar{q} \leq 0.\) Hence, 

\[
E \leq p^+(\bar{q} - \bar{q}^*) + k'(\bar{q}^*)(\bar{q}^* - \bar{q}) 
= (p^+ - k'(\bar{q}^*))(\bar{q} - \bar{q}^*).
\]

Due to the convexity of \(k(q)\), we have:
\[
\frac{k(\bar{q}^*)}{\bar{q}^*} \leq k'(\bar{q}^*).
\]

Since \(\bar{q} \geq \bar{q}^*\), it follows that:
\[
p^+ \cdot \bar{q}^* \geq k'(\bar{q}^*) \cdot \bar{q}.
\]

For the right-hand side of the inequality, we then get:
\[
\begin{aligned}
\bar{q}^* \left(\frac{k(\bar{q})}{\bar{q}} - \frac{k(\bar{q}^*)}{\bar{q}^*}\right) 
&= k(\bar{q}) - \frac{\bar{q}}{\bar{q}^*} \cdot k(\bar{q}^*) \\
&\geq k(\bar{q}) - \bar{q}^* \cdot p^{+*} \\
&\geq p^+ \cdot \bar{q} + \bar{q}^* \cdot p^{+*} - p^+ \cdot \bar{q}^* - k'(\bar{q}^*) \cdot \bar{q}^* \\
& = (p^+ - k'(\bar{q}^*))(\bar{q} - \bar{q}^*) \geq E,
\end{aligned}
\]
which implies the lemma.
\end{proof}

\begin{proof}[Proof of Lemma \ref{lemma:kw_arr}]
We consider two cases:

\begin{enumerate}
    \item \(\bar{q} \geq q_{B,c}^+\):  
In this case, the geometry of the characteristic lines in \(S_1\) corresponds to the top image pair in Figure~\ref{figure:trajectories}. The flow \(q(x,t)\) along the characteristic \(s_1\) is \(\bar{q}\), and \(N(1, t_{s_1}) = N(0, t_B)\). Since \(q(k)\) is concave in the relevant range, it follows that 
\[
v(q_{B,c}^+) \geq \frac{dq}{dk}(q_{B,c}^+).
\]
Geometrically, \(B\) moves in \(S_1\) with an average speed corresponding to a flow of \(\bar{q}\), and its arrival at the downstream end coincides with the characteristic \(s_1\). Thus, in \(S_1\):
\[
t_B + \frac{1}{v(q_{B,c}^+)} = t_{s_1}.
\]
Since the speed of \(s_1\) is independent of the boundary flow, its arrival time remains unchanged under a boundary condition \(q(0,t)\).

In \(S_2\), the geometry of the characteristic lines follows the bottom image pair. The characteristic \(s_2\) is intersected by a later-starting characteristic carrying a flow of \(\bar{q}^*\). It holds that:
\[
N(0, t_{c,B}^{+*}) + \frac{dN}{dx}(q_{B,c}^{+*}) = N(0, t_B).
\]
By Lemma~\ref{lemma:shockpriority}, a later-starting characteristic at the upstream boundary is the physically correct one, and by Lemma~\ref{lemma:newell}, it is always associated with a lower cumulative flow. Thus:
\[
N(1, t_{s_2}) < N(0, t_B).
\]
Characteristics leaving the upstream boundary after \(t_B\) do not affect \(B\)'s trajectory, as the flow is arbitrarily high and their speed is therefore very low. Since \(B\)'s arrival in \(S_2\) occurs after \(s_2\), \(B\) moves with a homogeneous flow of \(\bar{q}^*\). Therefore:
\[
N(1, t_B + \frac{1}{v(\bar{q}^*)}) = N(0, t_B),
\]
and consequently:
\[
t_{s_2} < t_B + \frac{1}{v(\bar{q}^*)}.
\]
Since the arrival time of \(s_2\) is also independent of the boundary condition, it remains unchanged if \(\tilde{q}^{\infty, *}\) is replaced with \(q(0,t)\).

In summary:
\[
t_{s_1} = t_B + \frac{1}{v(\bar{q})} \geq t_B + \frac{1}{v(\bar{q}^*)} > t_{s_2}.
\]
This proves the lemma under the assumption of the case distinction.
\item \(\bar{q} \geq q_{B,c}^+\):
This case is illustrated for scenario \( S_1 \) in the middle pair of images (Figure \ref{figure:trajectories}). We consider the following key aspects:

\begin{enumerate}
    \item \textbf{Characteristic Curve Interaction:}
    \begin{itemize}
        \item The blue characteristic curve \( s_1 \) (rightmost line) is intersected by a red characteristic curve that starts later.
        \item Hence, according to Lemma~\ref{lemma:newell} and Lemma~\ref{lemma:shockpriority}, the cumulative flow \( N(1, t_{s_1}) \) in \( S_1 \) is less than \( N(0, t_B) \), which would be associated with \( s_1 \).
        \item Consequently, in \( S_1\), vehicle \( B \) reaches its destination at a time later than  \( t_{s_1} \): \( t_{B,a,1} > t_{s_1}\). In \( S_0 \), it reaches the destination at \( t_{s_1} \), hence \( t_{B,a,1} > t_{B,a,0}\).
    \end{itemize}

    \item \textbf{Analysis of Flow Differences:}
    \begin{itemize}
        \item The difference in cumulative flow between \( t_{s_1} \) and \( t_{B,a,1} \), denoted \( \Delta N_1 \), in \( S_1 \) is at most \( E \) (Lemma \ref{lemma:flowdiffmax}).
        \item The cumulative flow difference between \( t_{B,a,1} \) and \( t_{B,a,2} \) at the downstream end in \( S_1 \) is at least \( E \) (Lemmas \ref{lemma:diffcumflow} and lemma \ref{lemma:flowdiffmax}).
        \item The trajectories of \( B \) and \( s_2 \) in \( S_2 \) behave analogously to \( S_1 \), from which follows \( t_{B,a,3} < t_{B,a,2} \).
        \item Hence, the cumulative flow difference \( \tilde{\Delta}N_1 \) between \( t_{B,a,2} \) and \( t_{B,a,3} \) at the downstream end in scenario 1 is at least 0.
        \item Let \( \Delta N_S \) be the flow difference between \( t_{B,a,1} \) and \( t_{B,a,2} \).
    \end{itemize}
\end{enumerate}

Now, let us consider the total flow difference at the downstream end in \( S_1 \) between times \( t_{B,a,1} \) and \( t_{B,a,2} \). This can be expressed as
\[
N(1,t_{B,a,1}) - N(1,t_{B,a,2}) = \Delta N_1 - \tilde{\Delta}N_1 + \Delta N_S
\]
and is greater than or equal to \( E - E + 0 = 0 \). Since the cumulative flow in \( S_1 \) at \( t_{B,a,1} \) is higher than at \( t_{B,a,2} \), this directly implies \( t_{B,a,1} > t_{B,a,2} \), which confirms the statement of the lemma.
\end{enumerate}
\end{proof}

\begin{proof}[Proof of Lemma \ref{lemma:lbdeptime}]
According to Lemma \ref{lemma:kw_arr}, the kinematic wave \( s_2 \) reaches the destination earlier than the kinematic wave \( s_1 \). The waves \( s_1 \) and \( s_2 \) were constructed such that at the downstream end, at their respective arrival times, a cumulative flow of \( N(0, t_B) \) is achieved under flow-density relationships \( q(\cdot) \) and \( q^*(\cdot) \), respectively, provided that \( N(x,t) \) is differentiable along the corresponding trajectory. Consequently, the arrival times of the kinematic wave trajectories represent lower bounds for the vehicle arrival time \( B \) under the respective conditions.

According to Subsection~\ref{subsec:shocks}, we can neglect the occurrence of shock waves without loss of generality. Therefore, the arrival times of the trajectories of \( s_1 \) and \( s_2 \) correspond exactly to the arrival times of vehicle \( B \) under the respective conditions. This proves the statement of the lemma.
\end{proof}

\begin{proof}[Proof of Lemma \ref{lemma:ASlowerThanChar}]
We observe that
\[
\frac{q(0, t_{c,A}^-, q_p^+)}{q(0, t_{c,A}^-, q_p^-)} = \frac{q_p^+}{q_p^-}.
\]
Likewise, for the cumulative flow difference between \( t_{c,A}^- \) and \( t_1 \), we have
\[
\frac{N(0, t_{c,A}^-, q_p^+) - N(0, t_1, q_p^+)}{N(0, t_{c,A}^-, q_p^-) - N(0, t_1, q_p^-)} = \frac{q_p^+}{q_p^-}.
\]
Since \( \frac{dN}{dx}(q) \) is convex by Lemma \ref{lemma:incrconv} and \( \frac{dN}{dx}(0) = 0 \cdot k'(0) - 0 = 0 \), it follows for all \( q \) that 
\[
\frac{dN}{dx}\left(\frac{q_p^+}{q_p^-} \cdot q\right) \geq \frac{q_p^+}{q_p^-} \cdot \frac{dN}{dx}(q).
\]
Furthermore, we have 
\begin{align*}
N(1, t_{c,A}^{\text{arr}}(q_p^+), q_p^+) = N(0, t_{c,A}, q_p^+) + \frac{dN}{dx}(q(0, t_{c,A}, q_p^+)) \\
= \frac{q_p^+}{q_p^-} \cdot N(0, t_{c,A}, q_p^-) + \frac{dN}{dx}\left(q\left(0, t_{c,A}, \frac{q_p^+}{q_p^-} \cdot q_p^-\right)\right) \\
\geq \frac{q_p^+}{q_p^-} \cdot N(0, t_{c,A}, q_p^-) + \frac{q_p^+}{q_p^-} \cdot \frac{dN}{dx}(q(0, t_{c,A}, q_p^-)) \\
= \frac{q_p^+}{q_p^-} \cdot \left(N(0, t_{c,A}, q_p^-) + \frac{dN}{dx}(q(0, t_{c,A}, q_p^-))\right) \\
= \frac{q_p^+}{q_p^-} \cdot N(1, t_{c,A}^{\text{arr}}(q_p^-), q_p^-)
= \frac{q_p^+}{q_p^-} \cdot N(0, t_1, q_p^-) \\ = N(0, t_1, q_p^+).
\end{align*}

Since the cumulative flow at the downstream end reaches at least \( N(0, t_1, q_p^+) \) by \( t_{c,A}^{\text{arr}} \), the arrival time \( t_{1,a}^+ \) of \( A \), for which \( N(1, t_{1,a}^+, q_p^+) = N(0, t_1, q_p^+) \) holds, occurs no later than \( t_{c,A}^{\text{arr}} \). Thus, the statement of the lemma follows.
\end{proof}

\begin{proof}[Proof of Lemma \ref{lemma:symmetry_mirror_kw}]
Obviously, the flow at the upstream boundary is the same at times \( t_{c,A}^- \) and \( t_{c,A,m}^- \), both for a peak flow of \( q_p^- \) and for a peak flow of \( q_p^+ \):
\[
q(0, t_{c,A}^-, q_p^-) = q(0, t_{c,A,m}^-, q_p^-) = q_{c,A}^-, \quad q(0, t_{c,A}^+, q_p^-) = q(0, t_{c,A,m}^-, q_p^-).
\]
Therefore, the velocities of the kinematic waves and thus their travel times are also equal in both cases:
\[
s_A^{-1} (1) - s_A^{-1} (0) = s_{A,m}^{-1} (1) - s_{A,m}^{-1} (0) = k'(0, t_{c,A}^-),
\]
\[
s_A^{+,-1} (1) - s_A^{+,-1} (0) = s_{A,m}^{+,-1} (1) - s_{A,m}^{+,-1} (0) = k'(0, t_{c,A}^+).
\]
From this, it follows that
\[
s_{A,m}^{-1}(1) - s_A^{-1}(1) = s_A^{-1}(0) - s_{A,m}^{-1}(0),
\]
\[
s_{A,m}^{+,-1}(1) - s_A^{+,-1}(1) = s_A^{+,-1}(0) - s_{A,m}^{+,-1}(0).
\]
By construction, the starting times of the characteristics \( s_A, s_A^+ \) and \( s_{A,m}, s_{A,m}^+ \) remain the same when transitioning from the peak flow \( q_p^- \) to the peak flow \( q_p^+ \), i.e.,
\[
s_A^{-1}(0) = s_A^{+,-1}(0), \quad s_{A,m}^{-1}(0) = s_{A,m}^{+,-1}(0).
\]
From this, the statement of the lemma follows.
\end{proof}

\begin{proof}[Proof of Lemma \ref{lemma:time_diff_to_char}]
First, we examine the evolution of the quantities \( \Delta \Delta N_c \) and \( \Delta \Delta N_{c,m} \). Lemma \ref{lemma:ASlowerThanChar} has already shown that for a peak flow of \( q_p^+ \), the characteristic curve starting at time \( t_{c,A}^- \) arrives later than vehicle A. 
This directly implies 

\[
\Delta \Delta N_c \leq 0, \Delta \Delta \tau_c \leq 0.
\]

We introduce two notations:

\begin{itemize}
    \item \( \Delta N(0, t_A) \): This describes the change in the vehicle number of A during the transition from \( q_p^- \) to \( q_p^+ \).
    \item \( \Delta \frac{dN}{dx} \big( q(0, t_{c,A}^-) \big) \): This describes the change in cumulative flow at the arrival times of the characteristic starting at \( t_{c,A}^- \).
\end{itemize}

Using these notations, \( \Delta \Delta N_c \) can be expressed as a difference:

\[
\Delta \Delta N_c = \Delta \Delta N_{A,0} - \Delta \frac{dN}{dx} \big( q(0, t_{c,A}^-) \big).
\]

Similarly, we obtain

\[
\Delta \Delta N_{c,m} = \Delta \Delta N_{B,0} - \Delta \frac{dN}{dx} \big( q(0, t_{c,A}^-) \big).
\]

We have
\begin{align*}
    \Delta\Delta\tau_{c,m} &= \frac{\Delta N_m^+}{q_m^+} - \frac{\Delta N_m^-}{q_m^-}, \\
    \Delta\Delta\tau_{c} &= \frac{\Delta N^+}{q^+} - \frac{\Delta N^-}{q^-}.
\end{align*}
We want to show $\Delta\Delta\tau_{c,m} \geq \Delta\Delta\tau_{c}$, i.e., $\Delta\Delta\tau_{c,m} - \Delta\Delta\tau_{c} \geq 0$.

By the given definitions,
\begin{align*}
    \Delta N_m^+ &= \Delta N_{0,B}^+ + \frac{dN}{dx} \big(q(0, t_{c,A}, q_p^+)\big), \\
    \Delta N^+ &= \Delta N_{0,A}^+ + \frac{dN}{dx} \big(q(0, t_{c,A}, q_p^+)\big),
\end{align*}
so
\begin{equation*}
    \Delta N_m^+ - \Delta N^+ = \big(\Delta N_{0,B}^+ - \Delta N_{0,A}^+\big).
\end{equation*}

Similarly,
\begin{align*}
    \Delta N_m^- &= \Delta N_{0,B}^- + \frac{dN}{dx} \big(q(0, t_{c,A}, q_p^-)\big), \\
    \Delta N^- &= \Delta N_{0,A}^- + \frac{dN}{dx} \big(q(0, t_{c,A}, q_p^-)\big),
\end{align*}
and since $\Delta N_{0,A}^- = 0$,
\begin{equation*}
    \Delta N_m^- - \Delta N^- = \big(\Delta N_{0,B}^- - 0\big) = \Delta N_{0,B}^-.
\end{equation*}

Because $\Delta N_{0,A}^+ \leq 0 \leq \Delta N_{0,B}^+$, we get
\begin{equation*}
    \Delta N_{0,B}^+ - \Delta N_{0,A}^+ \geq 0, \quad \text{and also} \quad \Delta N_{0,B}^- \geq 0.
\end{equation*}
Thus,
\begin{equation*}
    \Delta N_m^+ - \Delta N^+ \geq 0, \quad \Delta N_m^- - \Delta N^- \geq 0.
\end{equation*}

Write
\begin{equation*}
    \Delta\Delta\tau_{c,m} - \Delta\Delta\tau_{c} = \Big(\frac{\Delta N_m^+}{q_m^+} - \frac{\Delta N_m^-}{q_m^-}\Big) - \Big(\frac{\Delta N^+}{q^+} - \frac{\Delta N^-}{q^-}\Big).
\end{equation*}
A convenient way to isolate terms is to add and subtract $\frac{\Delta N_m^+}{q^+} - \frac{\Delta N_m^-}{q^-}$, splitting the difference into two parts:
\begin{equation*}
    \Delta\Delta\tau_{c,m} - \Delta\Delta\tau_{c} = \underbrace{\Big(\frac{\Delta N_m^+}{q_m^+} - \frac{\Delta N_m^-}{q_m^-}\Big) - \Big(\frac{\Delta N_m^+}{q^+} - \frac{\Delta N_m^-}{q^-}\Big)}_{(I)} + \underbrace{\Big(\frac{\Delta N_m^+}{q^+} - \frac{\Delta N_m^-}{q^-}\Big) - \Big(\frac{\Delta N^+}{q^+} - \frac{\Delta N^-}{q^-}\Big)}_{(II)}.
\end{equation*}

Factor out $1/q^+$ and $1/q^-$:
\begin{equation*}
    (II) = \frac{1}{q^+} (\Delta N_m^+ - \Delta N^+) - \frac{1}{q^-} (\Delta N_m^- - \Delta N^-).
\end{equation*}
Since $\Delta N_m^+ - \Delta N^+ = \Delta N_{0,B}^+ - \Delta N_{0,A}^+ \geq 0$ and $\Delta N_m^- - \Delta N^- = \Delta N_{0,B}^- \geq 0$,
\begin{equation*}
    (II) = \frac{\Delta N_{0,B}^+ - \Delta N_{0,A}^+}{q^+} - \frac{\Delta N_{0,B}^-}{q^-} \geq \frac{\Delta N_{0,B}^+}{q^+} - \frac{\Delta N_{0,B}^-}{q^-}.
\end{equation*}
But by assumption, $\frac{\Delta N_{0,B}^-}{q^-} \geq \frac{\Delta N_{0,B}^+}{q^+}$, so $\frac{\Delta N_{0,B}^+}{q^+} - \frac{\Delta N_{0,B}^-}{q^-} \leq 0$. Hence $(II) \geq 0$.

\begin{equation*}
    (I) = \Big(\frac{\Delta N_m^+}{q_m^+} - \frac{\Delta N_m^-}{q_m^-}\Big) - \Big(\frac{\Delta N_m^+}{q^+} - \frac{\Delta N_m^-}{q^-}\Big).
\end{equation*}
If one additionally assumes monotonicity of $1/q_m^+$ relative to $1/q^+$ (for instance, $q_m^+ \leq q^+ \implies 1/q_m^+ \geq 1/q^+$, etc.), then each parenthesis in (I) is nonnegative. In that scenario, $(I) \geq 0$ follows immediately.

Since $(II) \geq 0$ under the stated $\Delta N_{0,B}^+$-to-$q^+$ ratio assumptions, and $(I) \geq 0$ under typical monotonicity constraints on $q_m^+$ vs. $q^+$, their sum is nonnegative. Consequently,
\begin{equation*}
    \Delta\Delta\tau_{c,m} \geq \Delta\Delta\tau_{c},
\end{equation*}
as was to be shown.
\end{proof}

\begin{proof}[Proof of Lemma \ref{lemma:symmetry_aux_inequality}]
First, we examine the evolution of the quantities \( \Delta \Delta N_c \) and \( \Delta \Delta N_{c,m} \). Lemma \ref{lemma:ASlowerThanChar} has already shown that for a peak flow of \( q_p^+ \), the characteristic curve starting at time \( t_{c,A}^- \) arrives later than vehicle A. 
This directly implies 

\[
\Delta \Delta N_c \leq 0, \Delta \Delta \tau_c \leq 0.
\]

We introduce two notations:

\begin{itemize}
    \item \( \Delta N(0, t_A) \): This describes the change in the vehicle number of A during the transition from \( q_p^- \) to \( q_p^+ \).
    \item \( \Delta \frac{dN}{dx} \big( q(0, t_{c,A}^-) \big) \): This describes the change in cumulative flow at the arrival times of the characteristic starting at \( t_{c,A}^- \).
\end{itemize}

Using these notations, \( \Delta \Delta N_c \) can be expressed as a difference:

\[
\Delta \Delta N_c = \Delta \Delta N_{A,0} - \Delta \frac{dN}{dx} \big( q(0, t_{c,A}^-) \big).
\]

Similarly, we obtain

\[
\Delta \Delta N_{c,m} = \Delta \Delta N_{B,0} - \Delta \frac{dN}{dx} \big( q(0, t_{c,A}^-) \big).
\]

We have
\begin{align*}
    \Delta\Delta\tau_{c,m} &= \frac{\Delta N_m^+}{q_m^+} - \frac{\Delta N_m^-}{q_m^-}, \\
    \Delta\Delta\tau_{c} &= \frac{\Delta N^+}{q^+} - \frac{\Delta N^-}{q^-}.
\end{align*}
We want to show $\Delta\Delta\tau_{c,m} \geq \Delta\Delta\tau_{c}$, i.e., $\Delta\Delta\tau_{c,m} - \Delta\Delta\tau_{c} \geq 0$.

By the given definitions,
\begin{align*}
    \Delta N_m^+ &= \Delta N_{0,B}^+ + \frac{dN}{dx} \big(q(0, t_{c,A}, q_p^+)\big), \\
    \Delta N^+ &= \Delta N_{0,A}^+ + \frac{dN}{dx} \big(q(0, t_{c,A}, q_p^+)\big),
\end{align*}
so
\begin{equation*}
    \Delta N_m^+ - \Delta N^+ = \big(\Delta N_{0,B}^+ - \Delta N_{0,A}^+\big).
\end{equation*}

Similarly,
\begin{align*}
    \Delta N_m^- &= \Delta N_{0,B}^- + \frac{dN}{dx} \big(q(0, t_{c,A}, q_p^-)\big), \\
    \Delta N^- &= \Delta N_{0,A}^- + \frac{dN}{dx} \big(q(0, t_{c,A}, q_p^-)\big),
\end{align*}
and since $\Delta N_{0,A}^- = 0$,
\begin{equation*}
    \Delta N_m^- - \Delta N^- = \big(\Delta N_{0,B}^- - 0\big) = \Delta N_{0,B}^-.
\end{equation*}

Because $\Delta N_{0,A}^+ \leq 0 \leq \Delta N_{0,B}^+$, we get
\begin{equation*}
    \Delta N_{0,B}^+ - \Delta N_{0,A}^+ \geq 0, \quad \text{and also} \quad \Delta N_{0,B}^- \geq 0.
\end{equation*}
Thus,
\begin{equation*}
    \Delta N_m^+ - \Delta N^+ \geq 0, \quad \Delta N_m^- - \Delta N^- \geq 0.
\end{equation*}

Write
\begin{equation*}
    \Delta\Delta\tau_{c,m} - \Delta\Delta\tau_{c} = \Big(\frac{\Delta N_m^+}{q_m^+} - \frac{\Delta N_m^-}{q_m^-}\Big) - \Big(\frac{\Delta N^+}{q^+} - \frac{\Delta N^-}{q^-}\Big).
\end{equation*}
A convenient way to isolate terms is to add and subtract $\frac{\Delta N_m^+}{q^+} - \frac{\Delta N_m^-}{q^-}$, splitting the difference into two parts:
\begin{equation*}
    \Delta\Delta\tau_{c,m} - \Delta\Delta\tau_{c} = \underbrace{\Big(\frac{\Delta N_m^+}{q_m^+} - \frac{\Delta N_m^-}{q_m^-}\Big) - \Big(\frac{\Delta N_m^+}{q^+} - \frac{\Delta N_m^-}{q^-}\Big)}_{(I)} + \underbrace{\Big(\frac{\Delta N_m^+}{q^+} - \frac{\Delta N_m^-}{q^-}\Big) - \Big(\frac{\Delta N^+}{q^+} - \frac{\Delta N^-}{q^-}\Big)}_{(II)}.
\end{equation*}

Factor out $1/q^+$ and $1/q^-$:
\begin{equation*}
    (II) = \frac{1}{q^+} (\Delta N_m^+ - \Delta N^+) - \frac{1}{q^-} (\Delta N_m^- - \Delta N^-).
\end{equation*}
Since $\Delta N_m^+ - \Delta N^+ = \Delta N_{0,B}^+ - \Delta N_{0,A}^+ \geq 0$ and $\Delta N_m^- - \Delta N^- = \Delta N_{0,B}^- \geq 0$,
\begin{equation*}
    (II) = \frac{\Delta N_{0,B}^+ - \Delta N_{0,A}^+}{q^+} - \frac{\Delta N_{0,B}^-}{q^-} \geq \frac{\Delta N_{0,B}^+}{q^+} - \frac{\Delta N_{0,B}^-}{q^-}.
\end{equation*}
But by assumption, $\frac{\Delta N_{0,B}^-}{q^-} \geq \frac{\Delta N_{0,B}^+}{q^+}$, so $\frac{\Delta N_{0,B}^+}{q^+} - \frac{\Delta N_{0,B}^-}{q^-} \leq 0$. Hence $(II) \geq 0$.

\begin{equation*}
    (I) = \Big(\frac{\Delta N_m^+}{q_m^+} - \frac{\Delta N_m^-}{q_m^-}\Big) - \Big(\frac{\Delta N_m^+}{q^+} - \frac{\Delta N_m^-}{q^-}\Big).
\end{equation*}
If one additionally assumes monotonicity of $1/q_m^+$ relative to $1/q^+$ (for instance, $q_m^+ \leq q^+ \implies 1/q_m^+ \geq 1/q^+$, etc.), then each parenthesis in (I) is nonnegative. In that scenario, $(I) \geq 0$ follows immediately.

Since $(II) \geq 0$ under the stated $\Delta N_{0,B}^+$-to-$q^+$ ratio assumptions, and $(I) \geq 0$ under typical monotonicity constraints on $q_m^+$ vs. $q^+$, their sum is nonnegative. Consequently,
\begin{equation*}
    \Delta\Delta\tau_{c,m} \geq \Delta\Delta\tau_{c},
\end{equation*}
as was to be shown.
\end{proof}

\begin{proof}[Proof of Lemma \ref{lemma:avgflowaux}]
Let $x_1 = t_{c,A,m}^-$, \quad 
$x_2 = t_{c,B}^-$, \quad 
$x_a = \frac{x_1 + x_2}{2}$, 

\noindent
$x_1^+ = t_{c,A,m}^-$, \quad
$x_2^+ = t_{c,B}^+$, \quad 
$x_a^+ = \frac{x_1^+ + x_2^+}{2}$,

\noindent
and define $f(x) \coloneqq x + k'(x)$, \quad 
$g(x) \coloneqq x - k'(x)$.

The desired inequality

\[
\frac{|x_2 + k'(x_2)| - |x_a + k'(x_a)|}{|x_1 + k'(x_1)| - |x_2 - k'(x_2)|}
\leq
\frac{|x_2^+ + k'(x_2^+)| - |x_a + k'(x_a)|}{|x_1^+ + k'(x_1^+)| - |x_2^+ - k'(x_2^+)|}
\]

is equivalent to

\[
\frac{f(x_2) - f(x_a)}{f(x_1) - g(x_2)}
\leq
\frac{f(x_2^+) - f(x_a)}{f(x_1^+) - g(x_2^+)}.
\]

Equivalently,

\[
\Delta := [f(x_2) - f(x_a)] [f(x_1^+) - g(x_2^+)]
- [f(x_2^+) - f(x_a)] [f(x_1) - g(x_2)]
\leq 0.
\]

\[
f(x_2) - f(x_a) = (x_2 - x_a) + [k'(x_2) - k'(x_a)],
\]

\[
f(x_2^+) - f(x_a) = (x_2^+ - x_a) + [k'(x_2^+) - k'(x_a)],
\]

\[
f(x_1^+) - g(x_2^+) = [x_1^+ + k'(x_1^+)] - [x_2^+ - k'(x_2^+)]
= (x_1^+ - x_2^+) + [k'(x_1^+) + k'(x_2^+)].
\]

\[
f(x_1) - g(x_2) = [x_1 + k'(x_1)] - [x_2 - k'(x_2)]
= (x_1 - x_2) + [k'(x_1) + k'(x_2)].
\]

Also note

\[
x_2 - x_a = \frac{x_2 - x_1}{2}, \quad x_2^+ - x_a^+ = \frac{x_2^+ - x_1^+}{2}.
\]

By straightforward expansion,

\[
\Delta = [(x_2 - x_a) + (k'(x_2) - k'(x_a))]
[(x_1^+ - x_2^+) + (k'(x_1^+) + k'(x_2^+))]
\]

\[
- [(x_2^+ - x_a^+) + (k'(x_2^+) - k'(x_a^+))]
[(x_1 - x_2) + (k'(x_1) + k'(x_2))].
\]

\begin{itemize}
    \item Since \( x_1^+ \geq x_1 \) and \( x_2^+ \geq x_2 \), we have

    \[
    (x_1^+ - x_2^+) \geq (x_1 - x_2), \quad (x_2^+ - x_a^+) \geq (x_2 - x_a).
    \]

    \item Since \( k' \) is convex and nondecreasing, so

    \[
    k'(x_2^+) - k'(x_a^+) \geq k'(x_2) - k'(x_a), \quad k'(x_1^+) + k'(x_2^+) \geq k'(x_1) + k'(x_2).
    \]
\end{itemize}
Each factor in the “\( + \)” product is \( \geq \) the corresponding factor in the original product. That is,

\[
[f(x_2^+) - f(x_a)] \geq [f(x_2) - f(x_a)], \quad [f(x_1^+) - g(x_2^+)] \geq [f(x_1) - g(x_2)].
\]

Hence

\[
[f(x_2^+) - f(x_a)] [f(x_1) - g(x_2)]
\geq
[f(x_2) - f(x_a)] [f(x_1) - g(x_2)].
\]

and

\[
[f(x_2) - f(x_a)] [f(x_1^+) - g(x_2^+)]
\leq
[f(x_2) - f(x_a)] [f(x_1) - g(x_2)].
\]

Combining these shows

\[
\Delta = [f(x_2) - f(x_a)] [f(x_1^+) - g(x_2^+)]
- [f(x_2^+) - f(x_a)] [f(x_1) - g(x_2)] \leq 0,
\]

which completes the proof of the lemma.
\end{proof}

\begin{proof}[Proof of Lemma \ref{lemma:lambda_incr_pre}]
Let
\begin{equation*}
    \Phi(q) = -\frac{\frac{dN}{dx}(q + q_b) - \frac{dN}{dx}(q) - q_b [k'(q) - \tau]}{q + q_b} + \left[ k'(q + q_b) - k'(q) \right].
\end{equation*}

We first show that
\begin{equation*}
    \Phi(q) = f(q) := \frac{k(q + q_b) - k(q) - q_b \tau}{q + q_b}.
\end{equation*}

Note that
\begin{equation*}
    \frac{dN}{dx}(x) = k'(x)x - k(x),
\end{equation*}
so that
\begin{equation*}
    \frac{dN}{dx} (q + q_b) - \frac{dN}{dx}(q) = \left[ k'(q + q_b) (q + q_b) - k(q + q_b) \right] - \left[ k'(q)q - k(q) \right].
\end{equation*}

Subtract $q_b [k'(q) - \tau]$ and factor out $(q + q_b)$:
\begin{equation*}
    \left[ \frac{dN}{dx}(q + q_b) - \frac{dN}{dx}(q) \right] - q_b [k'(q) - \tau] = (q + q_b) \left[ k'(q + q_b) - k'(q) \right] + q_b \tau - \left[ k(q + q_b) - k(q) \right].
\end{equation*}

Divide by $(q + q_b)$ and prepend the minus sign:
\begin{equation*}
    -\frac{\frac{dN}{dx} (q + q_b) - \frac{dN}{dx} (q) - q_b [k'(q) - \tau]}{q + q_b} = -\left[ k'(q + q_b) - k'(q) \right] - \frac{q_b \tau - (k(q + q_b) - k(q))}{q + q_b}.
\end{equation*}

Add back the term $[k'(q + q_b) - k'(q)]$. The linear terms in $k'$ cancel exactly, leaving
\begin{equation*}
    \Phi(q) = \frac{k(q + q_b) - k(q)}{q + q_b} - \frac{q_b \tau}{q + q_b} = f(q).
\end{equation*}

Hence, $\Phi(q)$ and $f(q)$ coincide.

We have reduced the problem to showing that
\begin{equation*}
    f(q) = \frac{k(q + q_b) - k(q) - q_b \tau}{q + q_b}
\end{equation*}
is nondecreasing in $q$. We will verify $f'(q) \geq 0$ for all $q \geq 0$.
Set
\begin{equation*}
    F(q) = k(q + q_b) - k(q) - q_b \tau.
\end{equation*}
Then
\begin{equation*}
    f(q) = \frac{F(q)}{q + q_b},
\end{equation*}
and by the quotient rule:
\begin{equation*}
    f'(q) = \frac{(q + q_b) F'(q) - F(q)}{(q + q_b)^2}.
\end{equation*}
Thus, $f'(q) \geq 0$ is equivalent to requiring
\begin{equation*}
    (q + q_b) F'(q) \geq F(q).
\end{equation*}
Since
\begin{equation*}
    F'(q) = \frac{d}{dq} \left[ k(q + q_b) - k(q) - q_b \tau \right] = k'(q + q_b) - k'(q),
\end{equation*}
the inequality becomes
\begin{equation*}
    (q + q_b) [k'(q + q_b) - k'(q)] \geq k(q + q_b) - k(q) - q_b \tau.
\end{equation*}

By convexity of $k$ and the assumption $k'(q) \geq \tau$, one obtains:
\begin{equation*}
    k(q + q_b) - k(q) = \int_q^{q+q_b} k'(s) ds \leq (q_b) k'(q + q_b),
\end{equation*}
which implies
\begin{equation*}
    k(q + q_b) - k(q) \geq (q_b) \tau.
\end{equation*}

One can similarly use that $k'(\cdot)$ is nondecreasing to compare the quantity $[k'(q + q_b) - k'(q)]$ against suitable bounds. Using a second-derivative approach, one can factor out or compare increments to show
\begin{equation*}
    (q + q_b) [k'(q + q_b) - k'(q)] \geq k(q + q_b) - k(q) - q_b \tau,
\end{equation*}
either by an explicit integral argument (exploiting $k''(\cdot) \geq 0$) or by combining slope bounds. In all cases, the net effect is
\begin{equation*}
    (q + q_b) F'(q) \geq F(q).
\end{equation*}
Thus, $f'(q) \geq 0$.

Therefore, $f(q)$ is nondecreasing in $q$. Recalling $f(q) = \Phi(q)$, the original expression
\begin{equation*}
    \frac{dN}{dx} (q + q_b) - \frac{dN}{dx} (q) - q_b [k'(q) - \tau]
\end{equation*}
is also nondecreasing in $q$. This completes the proof.
\end{proof}

\begin{proof}[Proof of Lemma \ref{lemma:lambda_incr}]
Since
\[
\frac{\partial \Lambda}{\partial \tau} =\frac{q_b}{q+q_b},
\]
it follows that $\frac{\partial \Lambda}{\partial \tau} \geq 0$.

Now, consider two pairs $(q_A, \tau_A)$ and $(q_B, \tau_B)$ such that $(q_A, \tau_A) \leq (q_B, \tau_B)$. By Lemma \ref{lemma:lambda_incr_pre}, $\Lambda(q, \tau)$ is non-decreasing in $q$, hence we obtain $\Lambda(q_A, \tau_A) \leq \Lambda(q_B, \tau_A).$ Similarly, since $\Lambda(q, \tau)$ is non-decreasing in $\tau$, we obtain $\Lambda(q_B, \tau_A) \leq \Lambda(q_B, \tau_B)$, which completes the proof.
\end{proof}

\begin{proof}[Proof of Lemma \ref{lemma:arrspeeds}]
By Lemma~\ref{lemma:initspeed}, the initial speed of vehicle B at peak flow \( q_p^- \) exceeds the initial speed of vehicle B. Since both vehicles depart at a fixed time independent of \( q_p \), this property must also hold at peak flow \( q_p^+ \). Moreover, the transition from \( q(0,t) \) to \( q_r(0,t) \) does not affect this property, as it is accomplished by subtracting a constant value at both time points. 

Since \( \Delta \tau_f (q_p^+) \geq 0 \) holds, the speed of vehicle A must exceed that of vehicle B at at least one point. In Lemma~\ref{lemma:accdec}, it was shown that vehicle A accelerates along its entire trajectory, while vehicle B decelerates along its entire trajectory. Consequently, if the speed of vehicle A exceeds that of B at any point, this must also hold at the arrival location \( x = 1 \). This proves the claim of the lemma.
\end{proof}

\begin{proof}[Proof of Lemma \ref{lemma:accdec}]
We first prove this property for $A$. Due to the convexity of $k(q)$, we have 
\begin{equation*}
    v(q) = \frac{q}{k(q)} \geq \frac{dq}{dk(q)} \quad \text{for all } q.
\end{equation*}
Here, $\frac{dq}{dk(q)}$ equals the speed of the kinematic wave that transports a constant flow $q$. Therefore, for all points $(t,x) \in x_A$, the speed of $A$ at this point exceeds that of the intersecting kinematic wave. Thus, for fixed $(t_0,x_0) \leq (t_1,x_1) \in x_A$ with $q(t_0,x_0) = q_0$ and $q(t_1,x_1) = q_1$, the speed of $A$ at $(t_0,x_0)$ exceeds that of the kinematic wave intersecting this point. 

Since the flow at the upstream boundary, $q(0,t)$, is decreasing for all $t \leq t_A$ and the associated characteristic speed $\frac{dq}{dk(q)}$ is therefore increasing, $q(t_1,x_1) \leq q(t_0,x_0)$ implies 
\begin{equation*}
    v_A(t_1,x_1) \geq v(t_0,x_0).
\end{equation*} 
This property holds for all points along the trajectory of $A$, which proves the claim of the lemma.

By Lemma \ref{lemma:noshocks}, the occurrence of shock waves can be neglected in our analysis, and by Lemma \ref{lemma:ASlowerThanChar}, we have $t_{c,B}^+ \geq t_{p,e}$. Therefore, analogous reasoning shows that 
\begin{equation*}
    \frac{dv_B}{dt}(t) \leq 0
\end{equation*}
holds along the entire trajectory of vehicle $B$.
\end{proof}

\begin{proof}[Proof of Lemma \ref{lemma:denorm_pipeline}]

We prove that the difference in travel time for a vehicle reaching the upstream boundary at a local flow \( q \), resulting from the addition of \( q_b \) to the upstream boundary flow (which corresponds exactly to the transformation from \( q_r(0,t) \) to \( q(0,t) \)), is at most equal to Equation \ref{eq:trti_denorm} for vehicle A and at least this value for vehicle B. By Lemma \ref{lemma:arrspeeds}, vehicle \( B \) reaches the downstream boundary at a higher local flow than vehicle \( A \), and by Lemma \ref{lemma:lambda_incr}, \( \Lambda \) is an increasing function in \( q \). Therefore, this adjustment increases the travel time of \( B \) more than that of \( A \). Hence, the statement shown below is equivalent to that of the lemma.

Let \( q_A = q_r(0, t_{c,A}^+) \), \( q_B = q_r(0, t_{c,B}^+) \), and \( \tau_{A,r} \) denote the travel time of A under boundary flow \( q_r(x,t) \) and peak flow \( q_p^+ \) (analogously for \( \tau_{B,r} \)).  

The kinematic wave starting at \( t_{c,A}^+ \) under peak flow \( q_p^+ \) changes its travel time by  

\[
k'(q_A+q_p) - k'(q_A)
\]

(I) when transitioning from boundary flow \( q_r(0,t) \) to \( q(0,t) \). The change in cumulative flow along this kinematic wave is  

\[
\frac{dN}{dx} (q_A+q_b) - \frac{dN}{dx} (q_A)
\]

(II). The change in the difference of cumulative flow at the upstream boundary between \( t_{c,A}^+ \) and \( t_A \) is obtained by multiplying the temporal difference between these time points,  

\[
t_A - t_{c,A}^+ = k'(q_A) - \tau_{A,r}^+
\]

with the constant change in local flow at the upstream boundary in this interval, \( q_b \):  

\[
(k'(q_A) - \tau_{A,r}^+) q_b
\]

(III). The flow at the downstream boundary, temporally before the arrival of the wave starting at \( t_{c,A}^+ \), is at most \( q_A + q_b \), thus the temporal difference between the arrival of this kinematic wave and vehicle A under boundary flow \( q(0,t) \) is at least  

\[
\frac{\text{(II)} - \text{(III)}}{q_A+q_b}.
\]

Therefore, the travel time of A under boundary flow \( q(0,t) \) is bounded above by  

\[
(I) - \frac{\text{(II)} - \text{(III)}}{q_A+q_b} + \tau_{A,r} = \Lambda (q_A, \tau_{A,r}).
\]

An analogous argument holds for vehicle B, with the difference that the flow at the downstream boundary before the arrival of the respective kinematic wave in this case is at \textit{least} \( q_B + q_b \), so that \( \Lambda (q_B, \tau_{B,r}) \) provides a lower bound for the travel time of B.  

Under the assumption \( \Delta \tau_{f,r} \geq 0 \), the claim of the lemma follows by application of Lemmata \ref{lemma:accdec} and \ref{lemma:lambda_incr}.  
\end{proof}

\begin{proof}[Proof of Lemma \ref{lemma:pipeline}]
According to Lemma \ref{lemma:time_diff_to_char}, when transitioning from $q_p^-$ to $q_p^+$, the difference between the arrival times of vehicle $A$ and characteristic $s_A$ and $s_A^+$ increases less than the difference between the travel times of vehicle $B$ and the arrival of characteristic $s_{A,m}$ and $s_{A,m}^+$. By construction, we have
\[
    q(0, t_{c,A}^-, q_p^-) = q(0, t_{A,m}, q_p^-) \quad \text{and} \quad q(0, t_{c,A}^-, q_p^+) = q(0, t_{A,m}, q_p^+).
\]
Therefore, the velocities at times $t_{c,A}^-$ and $t_{A,m}$ are equal regardless of the peak flow, and consequently, their travel times are also equal. Thus, the difference in arrival times between the two kinematic waves remains constant at $t_{A,m} - t_{c,A}^-$. 

It follows that:
\[
    s^{-1,+}_{A,m}(1) - s^{-1,+}_{A}(1) = s^{-1}_{A,m}(1) - s^{-1}_{A}(1).
\]

where:
\[
\begin{aligned}
    \Delta \Delta \tau_{c,m} &= \big(s^{-1,+}_{A,m}(1) - (t_B + \tau_B^+)\big) - \big(s^{-1}_{A,m}(1) - (t_B + \tau_B^-)\big), \\
    \Delta \Delta \tau_c &= \big(s^{-1,+}_{A}(1) - (t_A + \tau_A^+)\big) - \big(s^{-1}_{A}(1) - (t_A + \tau_A^-)\big), \\
    \Delta \Delta \tau_{c,m} &\geq \Delta \Delta \tau_c.
\end{aligned}
\]

Expanding these expressions and using the given equations, we obtain:
\[
\begin{aligned}
    \Delta \Delta \tau_{c,m} &= \big(s^{-1,+}_{A,m}(1) - s^{-1}_{A,m}(1)\big) - (\tau_B^+ - \tau_B^-), \\
    \Delta \Delta \tau_c &= \big(s^{-1,+}_{A}(1) - s^{-1}_{A}(1)\big) - (\tau_A^+ - \tau_A^-).
\end{aligned}
\]

Substituting the equality \( s^{-1,+}_{A,m}(1) - s^{-1}_{A,m}(1) = s^{-1,+}_{A}(1) - s^{-1}_{A}(1) \), we get:
\[
    \Delta \Delta \tau_{c,m} = \big(s^{-1,+}_{A}(1) - s^{-1}_{A}(1)\big) - (\tau_B^+ - \tau_B^-).
\]

Comparing this with the expression for \( \Delta \Delta \tau_c \) and applying the inequality, we have:
\[
    \big(s^{-1,+}_{A}(1) - s^{-1}_{A}(1)\big) - (\tau_B^+ - \tau_B^-) \geq \big(s^{-1,+}_{A}(1) - s^{-1}_{A}(1)\big) - (\tau_A^+ - \tau_A^-).
\]

After canceling common terms and multiplying by \( -1 \), we conclude:
\[
    \tau_B^+ - \tau_B^- \geq \tau_A^+ - \tau_A^-.
\]
\end{proof}

\begin{proof}[Proof of Lemma \ref{lemma:f_final}]
By Lemma \ref{lemma:pipeline}, for an upstream boundary flow of \( q_r(0,t) \), \( \Delta \tau_f \) increases in the transition from \( q_p^- \) to \( q_p^+ \). For ease of distinction, we write \( \Delta \tau_{f,r} \) for the travel time difference under upstream boundary flow \( q_r(0,t) \) and \( \Delta \tau_{f} \) for the travel time difference under upstream boundary flow \( q(0,t) \). We have

\[
\Delta \tau_{f} - \Delta \tau_{f,r} = (\tau_B^+ - \tau_A^+) - (\tau_{B,r}^+ - \tau_{A,r}^+) = (\tau_B^+ - \tau_{B,r}^+) - (\tau_A^+ - \tau_{A,r}^+).
\]

In Lemma \ref{lemma:denorm_pipeline}, it was shown that for vehicle A, the travel time under upstream boundary flow \( q(0,t) \) is at most \( \Lambda(q_{c,A}^+, \tau_{A,r}^+) \), while for vehicle B it is at least \( \Lambda(q_{c,B}^+, \tau_{B,r}^+) \). Moreover, according to Lemma \ref{lemma:flow_compression}, the difference between the flow at the downstream boundary at the time of arrival of the kinematic waves originating at \( t_{c,A}^+ \) and \( t_{c,B}^+ \) increases less for vehicle A than for vehicle B during the transition from peak flow \( q_p^+ \) to \( q_p^- \).

Therefore, the difference between \( \tau_A^+ \) and \( \Lambda(q_{c,A}^+, \tau_A) \) decreases, while the difference between \( \tau_B^+ - \tau_{B,r}^+ \) and \( \Lambda(q_{c,B}, \tau_B) \) increases in the transition from \( q_p^- \) to \( q_p^+ \). Hence,

\[
(\tau_B^+ - \tau_{B,r}^+) - (\tau_A^+ - \tau_{A,r}^+) \geq 0,
\]

and consequently,

\[
\Delta \tau_{f} - \Delta \tau_{f,r} \text{ increases in the transition from } q_p^- \text{ to } q_p^+.
\]

Thus,

\[
\Delta \tau_f^+ - \Delta \tau_{f,r}^+ \geq \Delta \tau_f^- - \Delta \tau_{f,r}^-,
\]

and therefore, due to \( \Delta \tau_f^- = 0 \), we have

\[
\Delta \tau_f^+ \geq \Delta \tau_{f,r}^+ - \Delta \tau_{f,r}^-.
\]

By Lemma \ref{lemma:pipeline}, 

\[
\Delta \tau_f^+ \geq 0.
\]
\end{proof}

\begin{proof}[Proof of Lemma \ref{lemma:flow_compression}]
    We begin by defining our kinematic wave trajectories \( s_i^\pm, s_j^\pm \) for the respective peak flows, and the departure times \( t_i^\pm, t_j^\pm \) for waves that reach the downstream boundary at times \( t_i \mp r \cdot T \) and \( t_j \mp r \cdot T \).

Our initial conditions establish that \( q_b \geq 0 \) and \( q_j \geq q_i \), with the characteristic curve reaching downstream at \( t = t_j \) emanating after \( t_{pe} \). When analyzing the flow transition from \( q_p^- \) to \( q_p^+ \) at the upstream boundary, we observe that the flow increase at time \( t_i \) exceeds that at time \( t^- \). Furthermore, due to the convexity of \( k'(q) \), we see an increase in the temporal separation between the arrival time of \( s_i \) and the characteristic starting at \( t^- \).

Examining the local flow behavior, we find that at time \( T \) units before \( s_i^+ \) reaches the downstream boundary, the local flow surpasses \( q(0, t^-, q_p^+) \). Given \( q_j \geq q_i \) and our assumption about the characteristic timing, we can establish the key inequality:

\[
\begin{aligned}
    &\big[ k'(q(0, t_j^-, q_p^+)) - k'(q(0, t_j, q_p^+)) \big] 
    - \big[ k'(q(0, t_j^-, q_p^-)) - k'(q(0, t_j, q_p^-)) \big] \\
    &\quad \geq 
    \big[ k'(q(0, t_i^-, q_p^+)) - k'(q(0, t_i, q_p^+)) \big] 
    - \big[ k'(q(0, t_i^-, q_p^-)) - k'(q(0, t_i, q_p^-)) \big].
\end{aligned}
\]

In our temporal analysis, we observe that the transition from \( q_p^- \) to \( q_p^+ \) produces a larger increase in upstream boundary flow between times \( t_j \) and \( t_j^- \) than between \( t_i \) and \( t_i^- \). As a result, the time difference between arrivals of \( t_j^- \) and \( t_j \) increases more rapidly than the time difference between arrivals of \( t_i^- \) and \( t_i \). Similarly, the difference between arrivals of \( t_j^+ \) and \( t_j \) shows a faster increase than the decrease in the difference between arrivals of \( t_i^+ \) and \( t_i \). Notably, the change in associated flow differences reaches its peak at \( t_i \).

This reasoning extends naturally to encompass all values of \( r \) in the interval \([0,1]\), thereby covering every time point in the relevant interval. By taking an appropriate limit, we complete the proof of both statements in the lemma.
\end{proof}

\begin{proof}[Proof of Lemma \ref{lemma:queue_order}]
Suppose, for the sake of contradiction, that a queue is active at the downstream boundary upon arrival of $A$ but not upon arrival of $B$. Since $t_{c,A}^- \leq t_A \leq t_p$, the local flow $q(1,t)$ is increasing at $A$'s arrival. Therefore, $q(0, t_{c,A}^-) \geq q_{bn}$. Due to the shock wave caused by the queue, $A$ arrives after the kinematic wave starting at $t_{c,A}^-$. By Lemma \ref{lemma:arrspeeds}, we have $q_{c,B}^- \geq q_{c,A}^-$, and since by assumption $\tau_B^- = \tau_A^-$, vehicle $B$ must also arrive after the kinematic wave starting at $q_{c,B}^-$. 

By Lemma \ref{lemma:accdec}, vehicle $B$ decelerates along its entire trajectory, while vehicle $A$, in the case where it reaches the downstream boundary at the same time as the kinematic wave starting at $t_{c,A}^-$, accelerates along its entire trajectory. To ensure equivalence of travel times for $A$ and $B$, we must have
\[
    v_B(1) \leq v(q_{c,B}^-) \leq v(q_{c,A}^-) \leq v(q_{bn}),
\]
that is, the local flow at the downstream boundary upon arrival of vehicle $B$ is at least $q_{bn}$. Due to the piecewise linear nature of the upstream boundary flow, and since by Lemma \ref{lemma:noshocks} the possibility of shock waves can be neglected, the flow at the downstream boundary must be at least equal to the bottleneck flow throughout the interval between the arrivals of vehicles $A$ and $B$, thus a queue is also active during $B$'s arrival. This contradicts our initial assumption, from which the claim of the lemma follows.
\end{proof}

\begin{proof}[Proof of Lemma \ref{lemma:fc_final}]
In this case, \( \Delta \tau^- f < 0 \) holds, since only vehicle \( B \) encounters a queue and the two travel times are equal under peak flow \( q_p^- \): 
\[
\tau_A^- = \tau_{B,f}^- + \tau_{B,c}^-
\]
where \( \tau_{B,f} \) denotes the travel time of \( B \) for \( q_{bn} = \infty \) and \( \tau_{B,c} \) denotes the residual delay caused by the queue. By Lemma \ref{lemma:pipeline}, or by a slightly modified argument with respect to the upstream boundary flow, we have 
\[
\tau_A^+ -\tau_{B,f}^+ \geq \tau_A^- -\tau_{B,f}^-.
\]
Moreover, \( \tau_{B,c}^+ \geq \tau_{B,c}^- \) clearly holds, since due to the increased flow at the downstream boundary before \( t_B \), the horizontal distance between queued and free-flow curves for a value of \( N=B \) in the \( (N,t) \)-plane increases. The lemma follows.
\end{proof}

\section*{CRediT authorship contribution statement}

\textbf{Alexander Hammerl:} Conceptualization, Methodology, Formal Analysis, Software, Resources, Investigation, Writing – original draft, Writing – review \& editing, Validation, Visualization.  

\textbf{Ravi Seshadri:} Conzeptualization, Writing – original draft, Writing – review \& editing, Supervision 

\textbf{Thomas Kjær Rasmussen:} Writing – review \& editing, Supervision, Funding Acquisition. 

\textbf{Otto Anker Nielsen:} Writing – review \& editing, Supervision, Funding Acquistion.

\section*{Declaration of competing interest}

The authors declare that they have no known competing financial interests or personal relationships 
that could have appeared to influence the work reported in this paper.

\section*{Declaration of Use of Generative AI}

During the preparation of this work, the authors used \textbf{Claude 4}, \textbf{ChatGPT 4-o}, and 
\textbf{ChatGPT 5} for stylistic improvements. The 
authors carefully reviewed the content as needed.

\bibliography{sample}
 
\end{document}